\pdfoutput=1
\documentclass[a4paper,UKenglish,cleverref,autoref,thm-restate,nolineno]{socg-lipics-v2021}

\usepackage{xspace}

\usepackage{booktabs}
\usepackage{makecell}

\usepackage{tikz}
\usetikzlibrary{calc,positioning}
\usepackage{pgfplots}
\pgfplotsset{compat=1.18}

\usepackage[ruled]{algorithm2e}
\SetKwComment{Comment}{/*}{*/}
\SetKw{Break}{break}
\SetKw{Continue}{continue}

\ccsdesc[300]{Theory of computation~Online algorithms}
\ccsdesc[300]{Theory of computation~Shared memory algorithms}
\ccsdesc[500]{Theory of computation~Data compression}
\ccsdesc[300]{Theory of computation~Sorting and searching}
\ccsdesc[500]{Mathematics of computing~Combinatorics on words}

\date{\today}
\author{Jannik Olbrich}{Ulm University, Germany}{jannik.olbrich@uni-ulm.de}{https://orcid.org/0000-0003-3291-7342}{}
\authorrunning{J.\ Olbrich}
\Copyright{Jannik Olbrich}
\keywords{Burrows-Wheeler Transform, Grammar compression}
\supplementdetails[subcategory={Source Code},swhid=swh:1:dir:0274cae4abb893b49bc035ef3a57986a418afe40]{Software}{https://gitlab.com/qwerzuiop/lyndongrammar}
\title{Fast and memory-efficient BWT construction of repetitive texts using Lyndon grammars}
\titlerunning{Efficient BWT construction using Lyndon grammars}
\hideLIPIcs

\EventEditors{Anne Benoit, Haim Kaplan, Sebastian Wild, and Grzegorz Herman}
\EventNoEds{4}
\EventLongTitle{33rd Annual European Symposium on Algorithms (ESA 2025)}
\EventShortTitle{ESA 2025}
\EventAcronym{ESA}
\EventYear{2025}
\EventDate{September 15--17, 2025}
\EventLocation{Warsaw, Poland}
\EventLogo{}
\SeriesVolume{351}
\ArticleNo{58}

\begin{document}

\let\originalleft\left
\let\originalright\right
\renewcommand{\left}{\mathopen{}\mathclose\bgroup\originalleft}
\renewcommand{\right}{\aftergroup\egroup\originalright}

\newcommand{\BWT}{\mathsf{BWT}\xspace}
\newcommand{\eBWT}{\mathsf{eBWT}\xspace}
\newcommand{\BBWT}{\mathsf{BBWT}\xspace}
\newcommand{\concBWT}{\mathsf{concBWT}\xspace}
\newcommand{\mdolBWT}{\mathsf{mdolBWT}\xspace}
\newcommand{\dolEBWT}{\mathsf{dolEBWT}\xspace}
\newcommand{\SA}{\mathsf{SA}\xspace}
\newcommand{\LCP}{\mathsf{LCP}\xspace}
\newcommand{\Ls}{\mathcal{L}\xspace}
\newcommand{\pss}{\mathrm{pss}\xspace}
\newcommand{\nss}{\mathrm{nss}\xspace}
\newcommand{\lexlt}{<_{\mathit{lex}}}
\newcommand{\lexgt}{>_{\mathit{lex}}}
\newcommand{\lexleq}{\leq_{\mathit{lex}}}
\newcommand{\lexgeq}{\geq_{\mathit{lex}}}
\newcommand{\omlt}{<_\omega}
\newcommand{\omgt}{>_\omega}
\newcommand{\omleq}{\leq_\omega}
\newcommand{\lce}[2]{\mathrm{lce}(#1, #2)}
\newcommand{\bra}[1]{\left(#1\right)}
\newcommand{\set}[1]{\left\{#1\right\}}
\newcommand{\ceil}[1]{\left\lceil#1\right\rceil}
\newcommand{\abs}[1]{\left|#1\right|}
\newcommand{\bigO}[1]{\mathcal{O}\left(#1\right)}
\newcommand{\intervCC}[2]{\left[#1\,..\,#2\right]}
\newcommand{\intervCO}[2]{\left[#1\,..\,#2\right)}
\newcommand{\intervOC}[2]{\left(#1\,..\,#2\right]}
\newcommand{\intervOO}[2]{\left(#1\,..\,#2\right)}
\newcommand{\sroot}{\operatorname{root}}
\newcommand{\period}{\operatorname{period}}
\newcommand{\val}[1]{\left[#1\right]}
\newcommand{\suff}[2]{\operatorname{suf}_{#2}\bra{#1}}
\newcommand{\todo}[1]{\textcolor{red}{\textsc{todo: }}\textcolor{gray}{#1}}

\newcommand{\cmsbwt}{\texttt{CMS-BWT}\xspace}
\newcommand{\libsais}{\texttt{libsais}\xspace}
\newcommand{\ropebwt}{\texttt{ropebwt3}\xspace}
\newcommand{\bigbwt}{\texttt{Big-BWT}\xspace}
\newcommand{\rpfbwt}{\texttt{r-pfbwt}\xspace}
\newcommand{\pfpebwt}{\texttt{PFP-eBWT}\xspace}
\newcommand{\grlbwt}{\texttt{grlBWT}\xspace}

\colorlet{factor_col}{brown!75!white!90!gray}

\maketitle
\begin{abstract}
	The Burrows-Wheeler Transform ($\BWT$) serves as the basis for many important sequence indexes.
	On very large datasets (e.g.\ genomic databases), classical $\BWT$ construction algorithms are often infeasible because they usually
	need to have the entire dataset in main memory.
	Fortunately, such large datasets are often highly repetitive.
	It can thus be beneficial to compute the $\BWT$ from a compressed representation.
	We propose an algorithm for computing the $\BWT$ via the Lyndon straight-line program, a grammar
	based on the standard factorization of Lyndon words.
	Our algorithm can also be used to compute the extended $\BWT$ ($\eBWT$) of a multiset of sequences.
	We empirically evaluate our implementation and find that we can compute the $\BWT$ and $\eBWT$ of very large datasets faster and/or with less memory than competing methods.
\end{abstract}

\section{Introduction}

The \emph{Burrows-Wheeler Transform} ($\BWT$) \cite{bwt} is one of the most important data structures in sequence analysis and is applied in a wide range of fields from data compression to bioinformatics.
It is obtained by assigning the $i$th symbol of the $\BWT$ to the last character of the $i$th lexicographically smallest conjugate of the input text.
The $\BWT$ can be computed in linear time and space.
The run-length compressed $\BWT$ ($\mathsf{RLBWT}$) uses the fact that the number of equal-character runs in the $\BWT$ is small compared to the text length when the text is repetitive \cite{bwt}.
Mantaci et al.\ \cite{ebwt} extended the notion of the $\BWT$ to string collections: the \emph{extended $\BWT$} ($\eBWT$) of a string collection consists of the last characters of the strings in $\mathcal M$ arranged in infinite periodic order (see Section~\ref{sec:preliminaries} for definitions).


The $\BWT$ serves as the basis of several \emph{self-indexes}, i.e., data structures supporting access to the text as well as pattern matching queries.
One of the most successful such indexes is the FM-index \cite{ferragina2005indexing}, which is used e.g.\ in the important bioinformatics programs BWA \cite{li2009fast} and Bowtie \cite{langmead2012fast}.
However, the FM-index uses space proportional to the size of the dataset and is thus not applicable for huge datasets that vastly exceed the size of main memory.
Such datasets often contain data from thousands or millions of individuals from the same species and are therefore extremely repetitive.
The \emph{r-index} \cite{gagie2018optimal} is based on the $\mathsf{RLBWT}$ and is the first text index of size proportional to the number $r$ of runs in the $\BWT$.
It can be constructed in space proportional to $r$ from the $\mathsf{RLBWT}$, which raises the problem of constructing the $\mathsf{RLBWT}$ of huge datasets with small working memory.
For this reason, various tools have emerged that exploit the repetitiveness of these datasets.

The Lyndon grammar of a text is a straight-line program (SLP), i.e., a context-free grammar in Chomsky normal form that generates a single text, where each symbol corresponds to a node in the Lyndon tree.
It was introduced in \cite{lyndon_grammar} as the basis for a self-index.
In this paper, we describe a method for efficiently computing several $\BWT$ variants of a text or a text collection via the input's Lyndon grammar and show that our method surpasses other tools with respect to time or memory used, especially when using multiple threads.


\subsection{Related work}
Boucher et al.\ \cite{boucher2019prefix} developed \emph{prefix-free parsing (PFP)},
a dictionary-compression technique that results in a dictionary and parse from which the ($\mathsf{RL}$)$\BWT$ can be constructed in space proportional to their total size
and in time linear to the size of the dataset.
While the dictionary typically grows slowly as the size of the dataset increases, the parse grows linearly with the input size (with a small constant chosen at runtime).
Consequently, for extremely large and repetitive datasets, ``the parse is more burdensome than the size of the dictionary'' \cite{oliva2023recursive}.
In an effort to remedy this, Oliva et al.\ \cite{oliva2023recursive} developed \emph{recursive prefix-free parsing}, which applies PFP to the parse.
PFP has also been used to construct the $\eBWT$ \cite{boucher2021computing}.
The currently fastest implementation for computing the suffix array $\SA$ (and $\BWT$) for general texts is \libsais and is based on
the linear-time \emph{suffix array induced sorting (SAIS)} algorithm \cite{nong2009linear,olbrich2024generic}.
\ropebwt uses \libsais to compute the $\BWT$ of chunks of the input and successively merges these $\BWT$s \cite{li2024bwt}.
D{\'\i}az-Dom{\'\i}nguez and Navarro used grammar compression to maintain the intermediate data in the SAIS algorithm with low main memory usage, resulting in a linear-time semi-external algorithm that computes a variant of the $\eBWT$ \cite{cenzato2024survey,diaz2023efficient}.
Specifically for a collection of sequences where each is similar to a given reference (e.g.\ human chromosomes), Masillo introduced \cmsbwt, which computes the $\BWT$ via \emph{matching statistics} \cite{masillo:LIPIcs.ESA.2023.83}.

Lyndon words have previously been used in the construction of $\SA$ and play an important role in the original $\eBWT$.
In \cite{MANTACI2013,MANTACI20142}, a strategy for computing $\SA$ was presented where ``local suffixes'' (i.e., substrings ending at boundaries between Lyndon factors) can be processed separately in each Lyndon factor.
Later, Baier effectively generalized the underlying idea by showing that sorting the suffixes by their longest Lyndon prefixes can be used for linear time $\SA$ construction \cite{baier2016linear}.
The \emph{bijective $\BWT$} ($\BBWT$) is the $\eBWT$ of the text's Lyndon factors \cite{scott_bbwt,kufleitner2009bijective,ebwt}.
Bonomo et al.\ showed that it is possible to compute the $\eBWT$ via the $\BBWT$ and gave an $\bigO{N\log N/\log\log N}$ construction algorithm based on properties of Lyndon words \cite{bonomo2013,bonomo2014}. Later, Bannai et al.\ modified the SAIS algorithm to compute the $\BBWT$ in linear time \cite{linearbbwt}.
Since publication, Baier's algorithm was improved in terms of time and memory usage \cite{bertram2021,olbrich2022} and further generalized such that it can be used to compute $\eBWT$ and $\BBWT$ in linear time \cite{olbrich2024generic}.
The principles used this generalized algorithm lie at the base of our algorithm for computing the $\BWT$, $\BBWT$ and $\eBWT$ variants from a Lyndon grammar.

\subsection{Our contributions}
Let $N$ be the size of a text or text collection, $n_\mathrm{max}$ the maximum size of a string in the text collection, and $g$ the size of the Lyndon SLP of the text or text collection.
We give two online in-memory algorithms for constructing the Lyndon SLP with $\bigO{Ng}$ and $\bigO{N\log g + g\log^2g}$ worst-case time complexity, both using $\bigO{g}$ words of memory.
Notably, the text or text collection is streamed from right to left and does not have to reside in main memory. 
Additionally, we give an expected linear-time algorithm that uses $\bigO{g + n_\mathrm{max}}$ words of memory.
Furthermore, we give an $\bigO{g}$ algorithm for sorting the symbols of a Lyndon SLP lexicographically by their generated strings,
and an $\bigO{N}$ algorithm for constructing the (run-length compressed) $\BWT$, $\BBWT$, or $\eBWT$ from a Lyndon SLP.
We implemented our algorithms and demonstrate empirically that we can construct the $\BWT$, $\BBWT$ or $\eBWT$ of repetitive texts or text collections faster and/or with less space than competing methods.


\section{Preliminaries}
\label{sec:preliminaries}

For $i,j\in\mathbb{N}_0$ we denote the set $\set{k\in\mathbb{N}_0 : i \leq k \leq j}$ by the
interval notations $\intervCC{i}{j} = \intervCO{i}{j+1} = \intervOC{i-1}{j} = \intervOO{i-1}{j+1}$.
For an array $\mathsf{A}$ we analogously denote the \emph{subarray} from $i$ to $j$ by
$\mathsf{A}\intervCC{i}{j} = \mathsf{A}\intervCO{i}{j+1} = \mathsf{A}\intervOC{i-1}{j} = \mathsf{A}\intervOO{i-1}{j+1} = \mathsf{A}[i]\mathsf{A}[i+1]\dots \mathsf{A}[j]$.
We use zero-based indexing, i.e., the first entry of the array $\mathsf{A}$ is $\mathsf{A}[0]$.

A \emph{string} $S$ of \emph{length} $n$ over an \emph{alphabet} $\Sigma$ is a
sequence of $n$ characters from $\Sigma$. We denote the length $n$ of $S$ by
$\abs{S}$ and the $i$th symbol of $S$ by $S[i-1]$, i.e., strings are zero-indexed.
In this paper we assume any string $S$ of length $n$ to be over a totally ordered
and linearly sortable alphabet (i.e., the characters in $S$ can be sorted in
$\mathcal{O}(n)$).
Analogous to arrays we denote the \emph{substring} from $i$ to $j$ by
$S\intervCC{i}{j} = S\intervCO{i}{j+1} =
S\intervOC{i-1}{j} = S\intervOO{i-1}{j+1} = S[i]S[i+1]\dots S[j]$.
For $j > i$ we let $S\intervCC{i}{j}$ be the \emph{empty string} $\varepsilon$.
For two strings $u$ and $v$ and an integer $k\geq 0$ we let $uv$ be the
concatenation of $u$ and $v$ and denote the $k$-times concatenation of $u$
by $u^k$.
A string $S$ is \emph{primitive} if it is non-periodic, i.e., $S=w^k$ implies $w=S$ and $k=1$.
The \emph{suffix} $i$ of a string $S$ of length $n$ is the substring $S\intervCO{i}{n}$ and is denoted by $\suff{S}{i}$.
Similarly, the substring $S\intervCC{0}{i}$ is a \emph{prefix} of $S$. A suffix (prefix) is \emph{proper} if
$i>0$ ($i + 1 < n$).
For two possibly empty strings $u$ and $v$, $uv$ is a conjugate of $vu$.

We assume totally ordered alphabets. This induces a total order on strings.
Specifically, we say a string $S$ of length $n$ is \emph{lexicographically smaller} than another string $T$ of length $m$ if and only if
there is some $\ell\leq\min\set{n,m}$ such that $S\intervCO{0}{\ell} = T\intervCO{0}{\ell}$ and either $n=\ell<m$ or $\ell <\min\set{n,m}$ and $S[\ell] < T[\ell]$, and
write $S\lexlt T$ in this case.
A non-empty string $S$ is in its \emph{canonical form} if and only if it is lexicographically minimal among its conjugates. If $S$ is additionally strictly lexicographically smaller than all of its other conjugates, $S$ is a \emph{Lyndon word}.
Equivalently, $S$ is a Lyndon word if and only if $S$ is lexicographically smaller than
all its proper suffixes \cite{duval1983}.

In the following, we use $\texttt{abbabcbcabb}$ as our running example.
\begin{theorem}[Chen-Fox-Lyndon theorem \cite{chen1958free}]
	\makeatletter\def\@currentlabelname{Chen-Fox-Lyndon theorem}\makeatother
	Any non-empty string $S$ has a unique \emph{Lyndon factorization}, that is, there is
	a unique sequence of Lyndon words (\emph{Lyndon factors}) $v_1\lexgeq\dots\lexgeq v_k$ with
	$S=v_1\dots v_k$.
	\label{theorem:chen_fox_lyndon}
\end{theorem}
The Lyndon factorization of our running example is $\texttt{abbabcbc},\texttt{abb}$.

\begin{definition}[Standard Factorization \cite{chen1958free}]
	The \emph{standard factorization} of a Lyndon word $w$ of length $\abs{w}\geq 2$ is the tuple $(u,v)$
	where $w=uv$ and $v$ is the longest proper suffix of $w$ that is Lyndon.
	The standard factorization $(u,v)$ of a Lyndon word $w$ with $|w|\geq 2$ always exists and both $u$ and $v$ are Lyndon.
	\label{def:standard_factorization}
\end{definition}
The standard factorization of the first Lyndon factor of our running example is $(\texttt{abb},\texttt{abcbc})$ as $\texttt{abcbc}$ is Lyndon and no longer proper suffix is Lyndon.

\begin{lemma}[{\cite[Lemma 3.20]{olbrich2024generic}}]
	\label{lemma:lyndon_children}
	Any Lyndon word $w$ with $\abs{w}\geq 2$ is of the form $w[0]w_{c_1}\cdots w_{c_k}$,
	where $w_{c_1}\lexgeq w_{c_2}\lexgeq \dots \lexgeq w_{c_k}$ are the Lyndon factors of $\suff{w}{1}$.
\end{lemma}

\begin{figure}
	\centering
	\foreach \c [count=\i from 0] in {a,b,b,a,b,c,b,c,a,b,b} {
		\expandafter\xdef\csname{chars\i}\endcsname{\noexpand\strut\noexpand\texttt{\c}}
	}
	\begin{tikzpicture}
		\foreach \x [count=\i from 0] in {abbabcbc,bbabcbca,babcbcab,abcbcabb,bcbcabba,cbcabbab,bcabbabc,cabbabcb,abb,bba,bab}{
			\ifthenelse{\i=0 \OR \i=8}{
				\expandafter\xdef\csname{excon\i}\endcsname{\noexpand\strut\noexpand\color{factor_col}\noexpand\texttt{\x}}
			} {
				\expandafter\xdef\csname{excon\i}\endcsname{\noexpand\strut\noexpand\texttt{\x}}
			}
		}
		\def\examplelen{11}
		
		\def\tablexoffset{1.25}
		\def\tableyoffset{-0.5}
		\node[rotate=90,text width=4.75cm,align=center] at (\tablexoffset,\tableyoffset-0.125-\examplelen*0.5*0.5) {\color{gray} Collect the conjugates of the Lyndon factors};
		\foreach \i in {0,1,2,3,4,5,6,7,8,9,10} {
			\node[anchor=north east] (conj\i) at (\tablexoffset+1,\tableyoffset-\i*0.5) {\strut\i};
			\node[anchor=north west] (conj\i) at (\tablexoffset+1,\tableyoffset-\i*0.5) {\csname{excon\i}\endcsname};
		}

		\draw[-stealth,very thick,draw=gray] (\tablexoffset+3,\tableyoffset-0.375) -- (\tablexoffset+3.5,\tableyoffset-0.375);
		\node[rotate=90,text width=4.75cm,align=center] at (\tablexoffset+3.25,\tableyoffset-0.125-\examplelen*0.5*0.5) {\color{gray} Sort them according to infinite periodic order};
		\draw[-stealth,very thick,draw=gray] (\tablexoffset+3,\tableyoffset-\examplelen*0.5+0.125) -- (\tablexoffset+3.5,\tableyoffset-\examplelen*0.5+0.125);

		\foreach [count=\i from 0] \s in {8,0,3,10,2,9,1,6,4,7,5} {
			\node[anchor=north east] (sorted\i) at (\tablexoffset+4.25,\tableyoffset-\i*0.5) {\strut\s};
			\node[anchor=north west] (sorted\i) at (\tablexoffset+4.25,\tableyoffset-\i*0.5) {\csname{excon\s}\endcsname};
		}

		\draw[-stealth,very thick,draw=gray] (\tablexoffset+6.25,\tableyoffset-0.375) -- (\tablexoffset+6.75,\tableyoffset-0.375);
		\node[rotate=90,text width=4.75cm,align=center] at (\tablexoffset+6.5,\tableyoffset-\examplelen*0.5*0.5) {\color{gray} Collect the last character of each entry};
		\draw[-stealth,very thick,draw=gray] (\tablexoffset+6.25,\tableyoffset-\examplelen*0.5+0.125) -- (\tablexoffset+6.75,\tableyoffset-\examplelen*0.5+0.125);

		\foreach [count=\i from 0] \s in {8,0,3,10,2,9,1,6,4,7,5} {
			\def\col{factor_col}
			\ifthenelse{\s=0} {
				\pgfmathtruncatemacro\ii{7}
			} {
				\ifthenelse{\s=8} {
					\pgfmathtruncatemacro\ii{10}
				} {
					\def\col{black}
					\pgfmathtruncatemacro\ii{round(\s-1)}
				}
			}
			\node[anchor=north] (bbwt\i) at (\tablexoffset+7.25,\tableyoffset-\i*0.5) {\color{\col}\csname{chars\ii}\endcsname};
		}

		\begin{scope}[xshift=2cm,yshift=2.25cm]
		\def\examplespacing{0.75}
		\foreach \s [count=\i from 0] in {8,2,3, 8,6,6,8,8, 11,10,11} {
			\node (idx\i) at (\examplespacing*\i,-0.75){\strut\i};
			\node (text\i) at (\examplespacing*\i,-1.25){\strut\csname{chars\i}\endcsname};
			\node (nss\i) at (\examplespacing*\i,-2.25){\strut\s};
		}
		\node (text11) at (\examplespacing*11,-1.25) {\strut};
		\foreach \s [count=\i from 0] in {8,2,3, 8,6,6,8,8, 11,10,11} {
			\draw[->,gray] ([yshift=0.2cm,xshift=0.1cm]text\i.south) to[out=340,in=200] ([yshift=0.2cm,xshift=-0.1cm]text\s.south);
		}
		\node[anchor=west] (idx) at (-1.2,-0.75) {\strut $i$};
		\node[anchor=west] (text) at (-1.2,-1.25) {\strut $S[i]$};
		\node[anchor=west] (pss) at (-1.2,-2.25) {\strut $\nss[i]$};
		\draw (-1, -2.65) -- (8, -2.65);
		\end{scope}
	\end{tikzpicture}
	\caption{
	Next smaller suffix array $\nss$ (top) and $\BBWT$ (bottom) of the running example $S=\texttt{abbabcbcabb}$.
	Each arrow points from $i$ to $\nss[i]$.
	Lyndon factors are coloured ({\color{factor_col}$\blacksquare$}).
	}
	\label{fig:intro_examples}
\end{figure}
Figure~\ref{fig:intro_examples} illustrates the following definitions.
\begin{definition}[next smaller suffix array, $\nss$]
	\label{def:nss}
	Let the \emph{next smaller suffix} array $\nss$ of a string $S$ be such that
	$\nss[i] = \min\set{j\in\intervOC{i}{\abs{S}} \mid \suff{S}{j} \lexlt \suff{S}{i}}$.
\end{definition}

\begin{definition}[Infinite Periodic Order]
	We write $S\omlt T$ if and only if the infinite concatenation $S^\infty = SS\dots$ is lexicographically smaller than the infinite concatenation $T^\infty=TT\dots$.
\end{definition}
For instance, $\texttt{ab}\lexlt\texttt{aba}\lexlt\texttt{abb}$ and $\texttt{abb}\omgt\texttt{ab}\omgt\texttt{aba}$ since $\texttt{abb}\dots\lexgt\texttt{abab}\dots\lexgt\texttt{abaab}\dots$.
\begin{definition}[Bijective Burrows-Wheeler Transform ($\BBWT$)]
	The \emph{bijective Burrows-Wheeler Transform ($\BBWT$)} of a string $S$ is the string obtained by taking the last character of each conjugate of the Lyndon factors of $S$ arranged in infinite periodic order.
\end{definition}
\begin{definition}[Extended Burrows-Wheeler Transform ($\eBWT$)]
	The \emph{extended Burrows-Wheeler Transform ($\eBWT$)} of a multiset $\mathcal M$ of strings is the string obtained by taking the last character of each conjugate of the strings in $\mathcal M$ arranged in infinite periodic order.
\end{definition}
Note that, by definition, the $\BBWT$ of a string $S$ is the $\eBWT$ of the Lyndon factors of $S$.
Similarly, the $\eBWT$ of a multiset $\mathcal M$ of strings is the
$\BBWT$ of the concatenation of the canonical forms of the strings in $\mathcal M$ arranged in lexicographically decreasing order \cite{linearbbwt} (the Lyndon factors of the resulting string are conjugates of the roots of the strings in $\mathcal M$).\footnote{
	The $\eBWT$ was originally defined for sets of primitive strings \cite{ebwt}. This limitation is unnecessary \cite{boucher2021computing,olbrich2024generic}.
}

For a primitive string $S$, its $\BWT$, the $\eBWT$ of $\{S\}$, and the $\BBWT$ of the smallest conjugate of $S$ are identical.
Consequently, the \emph{\$-$\BWT$} of a string $S$ (i.e., $\BWT(S\$)$, where $\$$ is smaller than all characters in $S$) commonly computed via the suffix array of $S$, is identical to $\BWT(\$S) = \BBWT(\$S) = \eBWT(\{\$S\})$.

\subsection{Lyndon Grammar}

\newcommand{\LTree}{\mathit{LTree}}
\begin{definition}[Lyndon Tree/Forest \cite{barcelo1990action,lyndon_grammar}]
	The \emph{Lyndon tree} of a Lyndon word $w$---denoted by $\LTree(w)$---is the ordered full binary tree defined recursively as follows:
	If $\abs{w}=1$, then $\LTree(w)$ consists of a single node labelled by $w$, and
	if $\abs{w}\geq 2$ and $w$ has the standard factorization $(u,v)$, the root of $\LTree(w)$ is labelled
	by $w$, the left child of the root is $\LTree(u)$, and the right child of the root is $\LTree(v)$.
	For a non-Lyndon word, we let the \emph{Lyndon Forest} be the sequence of Lyndon trees of the Lyndon factors.
	\label{def:lyndon_tree}
\end{definition}
The Lyndon Forest of the running example is shown in Figure~\ref{fig:tree_slp_example}.
Note that the Lyndon Forest of a string $S$ is closely related to the Lyndon array $\lambda$ of $S$ \cite{lyndon_grammar},
where $\lambda[i]$ is the length of the longest prefix of $\suff{S}{i}$ that is Lyndon (equivalently, $\lambda[i] = \nss[i] - i$ \cite{bille2019space}, cf.\ Figure~\ref{fig:intro_examples}).
A succinct representation of $\lambda$ occupies $2\abs{S}$ bits and can be computed in linear time directly from the text \cite{bille2019space,crochemore2020cartesian,lyndonarrayinducing}.

\begin{figure}
	\centering
	\foreach \c [count=\i from 0] in {a,b,b,a,b,c,b,c,a,b,b} {
		\expandafter\xdef\csname{chars\i}\endcsname{\noexpand\strut\noexpand\texttt{\c}}
	}
	\def\examplespacing{0.5}
	\newcommand{\treen}[4]{
		\draw let \p1 = (#2), \p2 = (#3) in node[inner sep=1pt,circle] (#1) at (\x1*0.5+\x2*0.5,{max(\y1,\y2)+0.7cm}) {\small{#4}};
		\draw[gray] (#1) -- (#2);
		\draw[gray] (#1) -- (#3)
	}
	\begin{tikzpicture}
		\begin{scope}
			\foreach \i in {0,1,2,3,4,5,6,7,8,9,10} {
				\node[inner sep=0pt,anchor=south] (c\i) at (\i*\examplespacing,0) {\strut\csname{chars\i}\endcsname};
			}
			\treen{n1}{c0}{c1}{$X_2$};
			\treen{n2}{n1}{c2}{$X_3$};
			\treen{n4}{c4}{c5}{$X_8$};
			\treen{n3}{c3}{n4}{$X_5$};
			\treen{n6}{c6}{c7}{$X_8$};
			\treen{n4}{n3}{n6}{$X_6$};
			\treen{n0}{n2}{n4}{$X_4$};
			\treen{n7}{c8}{c9}{$X_2$};
			\treen{n8}{n7}{c10}{$X_3$};
		\end{scope}
		\begin{scope}[xshift=6cm]
			\node[anchor=south west] (gr) at (0,0) {$\begin{aligned}
				X_1 &\to \texttt{a} & \quad X_6 &\to X_5X_8 \\
				X_2 &\to X_1X_7     &       X_7 &\to \texttt{b} \\
				X_3 &\to X_2X_7     &       X_8 &\to X_7X_9 \\
				X_4 &\to X_3X_6     &       X_9 &\to \texttt{c} \\
				X_5 &\to X_1X_8
			\end{aligned}$};
			\node[above=of gr,yshift=-1cm] {roots: $X_4$, $X_3$};
		\end{scope}
	\end{tikzpicture}
	\caption{
		Lyndon Forest (left) and Lyndon SLP (right) for the running example $S=\texttt{abbabcbcabb}$.
		For clarity, instead of the node labels the corresponding grammar symbols are shown.
		Note the structural similarity between the Lyndon forest and the arrows indicating $\nss$ in Figure~\ref{fig:intro_examples}.
	}
	\label{fig:tree_slp_example}
\end{figure}

\begin{definition}[Lyndon straight-line program \cite{lyndon_grammar}]
	\label{def:lyndon_slp}
	A \emph{Lyndon straight-line program} (SLP) is a context-free grammar $\mathcal G$ over an alphabet $\Sigma$ in Chomsky normal form, where $\val{X_i}$ is a Lyndon word for each symbol $X_i$ and each rule $X_i\to X_aX_b$ is such that the standard factorization of $\val{X_i}$ is $(\val{X_a},\val{X_b})$, where $\val{X_i}$ denotes the unique string generated by $X_i$.
	A symbol $X_i$ is called \emph{terminal symbol} if there is a rule of the form $X_i\to c$ ($c\in\Sigma$), and \emph{non-terminal symbol} otherwise.
	When $r_1,\dots,r_k$ are the root symbols of $\mathcal G$,\footnote{In slight deviation of the usual definition of SLPs, we allow multiple root/start symbols. The reason for this is that we want to represent arbitrary strings with the Lyndon SLP, not just Lyndon words.} $\mathcal G$ generates $\val{\mathcal G}=\val{r_1}\dots\val{r_k}$, and $\val{r_1}\geq\dots\geq\val{r_k}$ is the Lyndon factorization of $\val{\mathcal G}$.
	The \emph{size} $\abs{\mathcal G}$ of $\mathcal G$ is the number of production rules plus the number of start symbols. 
	The \emph{derivation tree} of the SLP $\mathcal G$ is a labelled ordered tree where the root node has the children $r_1,\dots,r_k$.
	We assume that distinct symbols generate distinct words.
\end{definition}
A Lyndon SLP arises from renaming the nodes of a Lyndon Forest such that two nodes with isomorphic subtrees (i.e., representing the same string)
are assigned the same symbol.
The Lyndon Forest and Lyndon SLP of our running example can be seen in Figure~\ref{fig:tree_slp_example}.

Clearly, the size of a Lyndon SLP $\mathcal G$ is bounded by $\bigO{\abs{\val{\mathcal G}}}$.
This bound is tight and there is no non-trivial bound in terms of e.g.\ the number of runs in the $\BWT$, the size $\gamma$ of the smallest string attractor \cite{kempa2018stringattractors} or the size of the smallest SLP:
consider $S=\texttt a^{n-1}\texttt b$ for $n>1$.
The smallest SLP generating $S$ has size $\bigO{\log n}$, its $\BWT$ has two runs, and the smallest string attractor has size $\gamma=\abs{\set{0,n-1}}=2$,
but every suffix of $S$ is Lyndon and thus the Lyndon SLP generating $S$ has $\Theta\bra{n}$ symbols.

In the following, we assume the set of symbols $V$ of an SLP $\mathcal G$ to be numbered consecutively, i.e., $V=\set{X_1,\dots,X_{\abs V}}$.

\section{Properties of sorted Lyndon Grammars}
\label{sec:lyndon_grammar_properties}

Consider a Lyndon grammar $\mathcal{G}$ where the symbols are sorted lexicographically,
i.e., $\val{X_i} \lexlt \val{X_j} \iff i < j$.
In this section we examine such Lyndon SLPs and find properties that enable sorting any Lyndon SLP lexicographically in linear time (Section~\ref{sec:sorting})
and give rise to an online algorithm for constructing Lyndon SLPs (Section~\ref{sec:btree_construction}).
Note that the grammar shown in Figure~\ref{fig:tree_slp_example} is lexicographically sorted.

\newcommand{\leftmost}{\mathcal{P}_L}
\newcommand{\rightmost}{\mathcal{P}_R}
\begin{definition}[Prefix- and Suffix-symbols]
	Define $\leftmost\bra{X_i}$ ($\rightmost\bra{X_i}$) as the set of symbols occurring on the leftmost (rightmost) path of $X_i$'s derivation tree.
	More formally, if $X_i$ is a terminal symbol define $\leftmost(X_i) = \rightmost\bra{X_i} = \set{X_i}$, and if $X_i\to X_aX_b$ is a non-terminal symbol define $\leftmost(X_i) = \set{X_i} \cup \leftmost(X_a)$ and $\rightmost\bra{X_i}=\set{X_i}\cup\rightmost\bra{X_b}$.
	We correspondingly define $\leftmost^{-1}(X_i)$ ($\rightmost^{-1}\bra{X_i}$) to be the set of symbols on whose leftmost (rightmost) path of the derivation tree $X_i$ occurs, i.e.\ $\leftmost^{-1}(X_i) = \set{X_j \mid X_i\in\leftmost\bra{X_j}}$ and $\rightmost^{-1}\bra{X_i} = \set{X_j \mid X_i\in\rightmost\bra{X_j}}$.
\end{definition}
For example, we have $\leftmost\bra{X_6} = \set{X_6,X_5,X_1}$, $\rightmost\bra{X_6} = \set{X_6,X_8,X_9}$ and $\leftmost^{-1}\bra{X_1} = \set{X_1,X_2,X_3,X_4,X_5,X_6}$ for the grammar shown in Figure~\ref{fig:tree_slp_example}.
\newcommand{\dset}{\mathcal{D}}
\begin{definition}
	Define $\dset(X_i)$ as the set of symbols occurring in $X_i$'s derivation tree.
	More formally, if $X_i$ is a terminal symbol define $\dset(X_i) = \set{X_i}$,
	and if $X_i\to X_aX_b$ is a non-terminal symbol define $\dset(X_i) = \set{X_i} \cup \dset(X_a) \cup \dset(X_b)$.
\end{definition}

Let $r(X_i)$ be the largest symbol whose derivation tree has $X_i$ on the leftmost path, i.e., $r(X_i) = \max\bra{\leftmost^{-1}\bra{X_i}}$.
Similarly, let $l(X_i)$ be the smallest symbol whose derivation tree has $X_i$ on the leftmost path, i.e.\ $l(X_i) = \min\bra{\leftmost^{-1}\bra{X_i}}$.
For example, $l\bra{X_1} = X_1$ and $r\bra{X_1} = X_6$ in the grammar from Figure~\ref{fig:tree_slp_example}.
In the remainder of this section, we show that all symbols $X_j$ in the interval $\intervCC{l(X_i)}{r(X_i)}$  generate strings with $\val{X_i}$ as prefix and even satisfy $X_j\in\leftmost\bra{X_i}$. This then implies that these intervals form a tree-structure.

\begin{lemma}
$l(X_i) = X_i$ for all $i$.
\end{lemma}
\begin{proof}
\label{lemma:left_bound}
For $X_j\in\leftmost^{-1}\bra{X_i}$ we have $\val{X_i}\lexleq\val{X_j}$ as $\val{X_i}$ is prefix of $\val{X_j}$ by definition.
\end{proof}

The following lemmas follow from basic properties of Lyndon words and the definition of the Lyndon grammar/tree.
\begin{restatable}{lemma}{lemmalyndonsuffix}
	\label{lemma:lyndon_suffix}
	Let $w = uv$ be a Lyndon word with standard factorization $(u,v)$. Then $\suff{u}{k}v=\suff{w}{k} \lexgt v$ holds for all $k\in[1..\abs{u})$.
\end{restatable}
\begin{proof}
	Assume $\suff{w}{k} \lexlt v=\suff{w}{\abs{u}}$ with $k\in\intervCO{1}{\abs{u}}$ for a contradiction.
	Consider the lexicographically smallest (not necessarily proper) suffix $\suff{w}{p}$ of $\suff{w}{k}$ (i.e., $k\leq p$).
	Since $\suff{w}{\abs u} = v$ is Lyndon, $\suff{w}{\abs u} \lexlt \suff{w}{\abs u + \ell}$ holds for $1\leq\ell<\abs{v}$ and we thus have $k \leq p \leq \abs{u}$.
	Since $\suff{w}{k} \lexlt v = \suff{w}{\abs{u}}$ holds by assumption, we additionally have $p \neq \abs{u}$ and consequently $k \leq p < \abs{u}$.
	Because $\suff{w}{p}$ is the lexicographically smallest suffix of $\suff{w}{k}$, it is also Lyndon.
	This is a contradiction because $v$ is by definition the longest Lyndon suffix of $w$.
\end{proof}

\begin{restatable}{lemma}{lemmacontiguous}
	\label{lemma:contiguous}
	For all $X_j\in \intervCC{l\bra{X_i}}{r\bra{X_i}}$ it holds $X_i \in \leftmost\bra{X_j}$, i.e., $X_i$ occurs on the leftmost path of $X_j$'s derivation tree.
	Equivalently, $\intervCC{l\bra{X_i}}{r\bra{X_i}} = \leftmost^{-1}\bra{X_i}$.
\end{restatable}
\begin{proof}
	Assume for a contradiction that there is some $X_j$ with $\val{X_i} \lexlt \val{X_j} \lexleq \val{r(X_i)}$ where $X_i\notin \leftmost(X_j)$.
	Because of simple properties of the lexicographical order and since $\val{r(X_i)}$ has $\val{X_i}$ as proper prefix, $\val{X_j}$ must also have $\val{X_i}$ as proper prefix.

	Therefore, there must be a symbol $X_r \to X_aX_b \in \leftmost(X_j)$ such that $\abs{\val{X_a}} < \abs{\val{X_i}} < \abs{\val{X_r}} \leq \abs{\val{X_j}}$,
	otherwise there would be a symbol generating $\val{X_i}$ in $\leftmost(X_j)$.

	Note that $\val{X_r} \lexleq \val{X_j} \lexlt \val{r(X_i)}$ because $\val{X_r}$ is a prefix of $\val{X_j}$,
	and $\val{X_i} \lexlt \val{X_r}$ because $\val{X_i}$ is a prefix of $\val{X_r}$ by definition.
	Also note that $r(X_i)$ is of the form $r(X_i)\to X_iX_c$ for some $c$.
	We now show the claim by deducing $\val{r(X_i)} \lexlt \val{X_r}$, which is a contradiction.

	Because $\val{X_i}$ is a proper prefix of $\val{X_r}$ and $\val{X_a}$ is a proper prefix of $\val{X_i}$, we can decompose $\val{X_r}$ into $wuv$ with
	$w = \val{X_a}, u = \suff{\val{X_i}}{\abs{\val{X_a}}}$ and $v = \suff{\val{X_b}}{\abs{u}}$.
	Clearly, $wu = \val{X_i}$ and $uv = \val{X_b}$.
	Because of Lemma~\ref{lemma:lyndon_suffix}, the standard factorization $(\val{X_i}, \val{X_c})$ of $\val{r(X_i)}$ implies $u\lexgt\val{X_c}$.
	From $uv = \val{X_b}$ we also have $v \lexgt u$ and thus $v\lexgt \val{X_c}$.
	This finally implies $\val{X_r} = wuv \lexgt wu\val{X_c} = \val{X_i}\val{X_c} = \val{r(X_i)}$.
\end{proof}
\begin{corollary}
	Any two intervals $[l(X_i),r(X_i)]$ and $[l(X_j),r(X_j)]$ either do not intersect or one is fully contained within the other.
	\label{cor:intervals_intersectionfree}
\end{corollary}
\begin{proof}
	Let $X_i$ and $X_j$ be different such that $l(X_i) < l(X_j) \leq r(X_i)$ (Lemma~\ref{lemma:left_bound} implies $l(X_i)\neq l(X_j)$).
	This implies $X_j\in\leftmost\bra{X_i}$ by Lemma~\ref{lemma:contiguous} and thus $\leftmost\bra{X_j}\subset\leftmost\bra{X_i}$ and $r(X_j) \leq r(X_i)$.
\end{proof}

\begin{figure}
	\centering
	\begin{tikzpicture}
		\newcommand{\tbox}[4]{ 
			\node[rectangle,draw,minimum width=#3*0.7cm,minimum height=0.45cm,anchor=north west,text width=#3*0.7cm-0.2cm,inner sep=1pt] at (#1*0.7cm,#2*-0.45cm) {\small{#4}}
		}
		\tbox{1}{0}{6}{$X_1$};
		\tbox{7}{0}{2}{$X_7$};
		\tbox{9}{0}{1}{$X_9$};
		\tbox{2}{1}{3}{$X_2$};
		\tbox{5}{1}{2}{$X_5$};
		\tbox{8}{1}{1}{$X_8$};
		\tbox{3}{2}{2}{$X_3$};
		\tbox{6}{2}{1}{$X_6$};
		\tbox{4}{3}{1}{$X_4$};

		\node[anchor=north west,inner sep=1pt] (gr) at (7.5,0) {\small$\begin{aligned}
			X_1 &\to \texttt{a} \quad& X_7  &\to \texttt{b} \quad & X_9 &\to \texttt{c} \\
			X_2 &\to X_1X_7          & X_5  &\to X_1X_8           & X_8 &\to X_7X_9 \\
			X_3 &\to X_2X_7          & X_6  &\to X_5X_8 \\
			X_4 &\to X_3X_6
		\end{aligned}$};
	\end{tikzpicture}
	\caption{First-symbol forest of the (lex.\ sorted) Lyndon SLP of the running example (cf.\ Figure~\ref{fig:tree_slp_example}).}
	\label{fig:interval_forest}
\end{figure}

\begin{corollary}
	\label{cor:symbol_tree}
	The intervals $[l(X_i),r(X_i)]$ induce a forest.
\end{corollary}
The following definition defines this forest. An example can be seen in Figure~\ref{fig:interval_forest}.
\begin{definition}
	We say an interval $[l(X_j),r(X_j)]$ is \emph{embedded} in an interval $[l(X_i),r(X_i)]$ if it is a subinterval of $[l(X_i),r(X_i)]$,
	i.e.\ $l(X_i) < l(X_j) \leq r(X_j) \leq r(X_i)$.
	If $[l(X_j),r(X_j)]$ is embedded in $[l(X_i),r(X_i)]$ and there is no interval embedded in $[l(X_i),r(X_i)]$ in which $[l(X_j),r(X_j)]$ is embedded,
	$[l(X_j),r(X_j)]$ is a \emph{child interval} of $[l(X_i),r(X_i)]$.

	The \emph{first-symbol forest} is a rooted forest, where each symbol $X_i$ of the SLP corresponds to a node $u_i$, and $u_j$ is a child of $u_i$ if and only if $[l(X_j),r(X_j)]$ is a child interval of $[l(X_i),r(X_i)]$.
	Note that the parent of $X_i\to X_aX_b$ is $X_a$.
	Also, the terminal symbols correspond to the roots of the trees in the forest.
\end{definition}

\begin{lemma}
	Let $X_i$ be a symbol of a Lyndon SLP.
	We have $
		\leftmost(X_i) = \set{ X_{\ell_0}, \dots, X_{\ell_k} }
	$
	where $X_{\ell_j} \to X_{\ell_{j-1}}X_{c_j}$ for all $j\in\intervCC{1}{k}$,
	$X_{\ell_0}$ is a terminal symbol, $X_{\ell_k} = X_i$, and $\val{X_{c_1}}\geq\dots\geq \val{X_{c_k}}$ is the Lyndon factorization of $\suff{\val{X_i}}{1}$.
	\label{lemma:lyndon_children_symbol}
\end{lemma}
\begin{proof}
	Follows immediately from Lemma~\ref{lemma:lyndon_children}.
\end{proof}

\section{Lexicographically sorting a Lyndon Grammar}
\label{sec:sorting}

Given a Lyndon SLP $\mathcal{G}$ with roots $r_1,\dots,r_k$, we want to
rename the symbols such that $\val{X_i} \lexlt \val{X_j}$ holds if and only if $i < j$.
This is possible in $\bigO{\abs{\mathcal{G}}}$ using a similar principle to what
is used for computing the Lyndon grouping in \cite{baier2016linear,olbrich2024generic}:
consider a symbol $X_i$. By Lemma~\ref{lemma:lyndon_children_symbol} and simple properties of Lyndon words, the symbols in
$\dset(X_i) \setminus \leftmost\bra{X_i}$ are lexicographically greater than $X_i$.
Therefore, when considering the symbols in lexicographically decreasing order
we can ``induce'' the position of a symbol $X_i\to X_aX_b$ upon encountering $X_b$:
the symbols in $\leftmost^{-1}\bra{X_i}$ must be the lexicographically largest symbols with $X_a$ as prefix symbol that have not yet been induced. 

Algorithm~\ref{alg:grammar_sorting} shows the procedure.
Throughout, $\mathsf A[i]$ is either the $i$th symbol in the lexicographical order or $\bot$, and for each inserted $X_i$, $\mathsf G[i]$ is the index in $\mathsf A$ of the smallest child of $X_i$ that has been inserted into $\mathsf A$ (or $r(X_i) + 1$ if no child has been inserted yet).
In the first for-loop, all terminal symbols are inserted into $\mathsf A$, and in the second for-loop, the remaining symbols are induced.
Consequently, after the second for-loop, $\mathsf A$ contains the lexicographic order of the symbols, and $\mathsf G$ is (almost) it's inverse, i.e.\ $\mathsf G[\mathsf A[i]] = i + 1$.

There are two operations in Algorithm~\ref{alg:grammar_sorting} that are not immediately obvious, namely finding $\abs{\leftmost^{-1}\bra{X_i}}$ and iterating over all $X_j$ with $X_j\to X_aX_{\mathsf A[i]}$ for a given $\mathsf{A}[i]$.
For the first, note that the values $\abs{\leftmost^{-1}\bra{X_i}}$ can be trivially computed in linear-time using the recurrence
$
\abs{\leftmost^{-1}\bra{X_a}} = 1 + \sum_{X_i\to X_aX_b}{\abs{\leftmost^{-1}\bra{X_i}}}
$.
Secondly, iterating efficiently over all $X_j$ with $X_j\to X_aX_{\mathsf{A}[i]}$
can be achieved using a linear-time preprocessing step where we collect for each symbol $X_i$ the symbols which have $X_i$ as second symbol on the right-hand side.

The following theorem follows.
\begin{restatable}{theorem}{thmlyndonsorting}
	Algorithm~\ref{alg:grammar_sorting} lexicographically sorts a Lyndon SLP in linear time.
	\label{theorem:grammar_sorting}
\end{restatable}
\begin{proof}
	\newcommand{\lpa}{\mathcal{PA}_\ell}
	Define $\lpa\bra{X_i} = \set{j \mid X_j\to X_iX_a }$.
	We prove by induction that the following invariants hold before each iteration of the second for-loop in Algorithm~\ref{alg:grammar_sorting}.
	The claim then follows immediately.
	Let $\pi$ be the (unique) permutation of $\intervCO0n$ such that $\val{X_{\pi[i]}} \lexlt \val{X_{\pi[j]}}$ if and only if $i < j$.
	After the second for-loop finishes, $\mathsf A = \pi$.
	Denote by $\mathcal I\subseteq\intervCO0n$ the set of integers in $\mathsf A$, i.e.\ the symbols already inserted.

	We claim that before the $k$-th iteration (i.e.\ $i=n-k$), the following holds:
	\begin{enumerate}[(1)]
		\item $\mathcal I = \set{\pi[j] \mid j\in\intervCO{i}{n}} \cup \set{j \mid X_j\to X_aX_b \text{ with } \pi^{-1}[b] > i} \cup \set{j\mid X_j \to c\in\Sigma}$
			\label{proof:set}
		\item $\mathsf A[j] \neq \bot \iff \mathsf A[j]\in\mathcal I$
			\label{proof:a}
		\item $\mathsf A[j] \neq \bot \iff \mathsf A[j] = \pi[j]$
			\label{proof:position}
		\item $\mathsf G[k] = \min\bra{\set{l(X_j) \mid j \in \mathcal I, X_j \to X_kX_b}\cup\set{r(X_k) + 1}}$ for each $k\in\mathcal I$
			\label{proof:end}
	\end{enumerate}

	\paragraph*{Induction base ($i = n-1$)}
	After the first for-loop, we have $\mathsf G[i] = r(X_i)$, $\mathsf G[i] = l(X_i) - 1$ and $\mathsf A[\mathsf G[i]] = i$ for each terminal symbol $X_i$.
	Since nothing else is written to $\mathsf A$, all invariants hold.

	\paragraph*{Induction step}
	Assume that the invariants hold for $i>0$. We show that after the iteration with $i$, the invariants hold for $i-1$.

	First note that $\mathsf A[i] = \pi[i]$ by Invariant~\ref{proof:set}.
	Observe that the inner for loop writes the element of the set $D := \set{j \mid X_j\to X_aX_{\mathsf A[i]}}$ to $\mathsf A$.
	For any $X_j\to X_aX_{\mathsf A[i]}$ (i.e.\ $j\in D$) we have $\val{X_j} \lexlt \val{X_{\mathsf A[i]}}$ because $\val{X_j}$ is a Lyndon word with $\val{X_{\mathsf A[i]}}$ as proper suffix.
	Thus, $X_j$ cannot be in the first and third sets used in Invariant~\ref{proof:set}.
	Since $\pi^{-1}[\mathsf A[i]] = i \not>i$, $j\notin\mathcal I$.
	This then implies that $D$'s elements have not been written to $\mathsf A$ in previous iterations (Invariant~\ref{proof:a}) and thus that the value at their correct position in $\mathsf A$ is still $\bot$ (Invariants~\ref{proof:a} and \ref{proof:position}).

	We now show that the elements of $D$ are written to the correct positions.
	First, it is clear that the first symbols on the right-hand side of the rules in $D$ are pairwise different (there are no two different rules generating the same string). Thus, we can consider each $j\in D$ individually and independently. Let $X_j\to X_aX_{\mathsf A[i]}$ with $j\in D$.
	By Invariant~\ref{proof:end}, there are two cases: $\mathsf G[a] = r(X_a)$ and $\mathsf G[a] = l(X_k)$ for some $k\in\mathcal I$ with $X_k\to X_aX_b$, and we need to show that $l(X_j) = \mathsf G[a] - \abs{\leftmost^{-1}\bra{X_j}}$ (because $l(X_j) = \pi^{-1}[j]$ by Lemma~\ref{lemma:left_bound}).
	Note that $X_j$ is the lexicographically largest element in $\lpa\bra{X_a}\setminus\mathcal I$ because any larger symbol with $X_a$ as first symbol must be in $\mathcal I$ by Invariant~\ref{proof:set}.
	From this, the claim then follows in both cases from Corollary~\ref{cor:intervals_intersectionfree} and Lemma~\ref{lemma:contiguous}.

	Corollary~\ref{cor:intervals_intersectionfree} then implies that $l(X_j) + \abs{\leftmost\bra{X_j}} - 1 = r(X_j)$ and thus that $\mathsf G[j]$ will be set correctly to $r(X_j) + 1$.

	We already showed that Invariants~\ref{proof:end} and \ref{proof:position} are maintained.
	Consider $\pi[i-1]$.
	If $X_{\pi[i-1]}$ is a terminal symbol, we have $\mathsf A[i-1] = \pi[i-1]$ by Invariants~\ref{proof:set} and \ref{proof:a}.
	Otherwise, $X_{\pi[i-1]}\to X_{a'}X_{b'}$ and $\val{X_{b'}}\lexgt \val{X_{\pi[i-1]}}$, i.e.\ $\pi^{-1}[b']\geq i$ and thus $\mathsf A[i-1] = \pi[i-1]$ again by Invariants~\ref{proof:set} and \ref{proof:a}.
	Therefore, Invariants~\ref{proof:set} and \ref{proof:a} are also maintained.
	\paragraph*{Time complexity}
	Note that in each iteration of the inner for-loop, a new element is inserted into $\mathcal I$.
	It follows that the number of iterations of all loops in Algorithm~\ref{alg:grammar_sorting} is bounded by $\bigO{g}$.
	Using the preprocessing steps explained earlier, the linear time complexity follows immediately.
\end{proof}

\begin{algorithm}
	\caption{Lexicographically sorting a Lyndon SLP}
	\label{alg:grammar_sorting}
	$(\mathsf A,s) \gets (\bot^n, 0)$\;
	\For(\tcc*[h]{in increasing lexicographical order of $c$}){$X_i\to c$ with $c\in\Sigma$}{
		$\mathsf{A}[s] \gets i$\tcc*[r]{There are $s$ symbols $X$ with $\val{X}[0]<c$}
		$s \gets s + \abs{\leftmost^{-1}\bra{X_i}}$\tcc*[r]{There are $\abs{\leftmost^{-1}\bra{X_i}}$ symbols $X$ with $\val{X}[0] = c$}
		$\mathsf{G}[i] \gets s$\tcc*[r]{$X_i$ is the lex.\ smallest symbol in $\leftmost^{-1}\bra{X_i}$}
	}
	\For{$i = n-1\to 0$}{
		\For{$X_j\to X_aX_{\mathsf{A}[i]}$}{
			\tcp{The symbols in $\leftmost^{-1}\bra{X_j}$ must be in $A\intervCO{G[a] - \abs{\leftmost^{-1}\bra{X_j}}}{G[a]}$}
			$\mathsf{G}[j]\gets\mathsf{G}[a]$\;
			$\mathsf{G}[a]\gets\mathsf{G}[a] - \abs{\leftmost^{-1}\bra{X_j}}$\;
			$\mathsf{A}[\mathsf{G}[a]]\gets j$\tcc*[r]{$X_j$ is the lex.\ smallest symbol in $\leftmost^{-1}\bra{X_j}$}
		}
	}
\end{algorithm}

\section{BWT construction}
\label{sec:bwt}

In this section, we describe an algorithm that derives the $\BBWT$ from the Lyndon SLP.
Specifically, combining the second phase of \texttt{GSACA} \cite{baier2016linear,olbrich2024generic} with run-length encoding on the lexicographically sorted Lyndon SLP results in a very efficient algorithm on real data.

The following lemmas and corollary establish a relationship between the lexicographical order of suffixes and the lexicographical order of certain symbols of the Lyndon grammar.
\begin{lemma}[{\cite[Lemma 3.3]{olbrich2024generic}}]
	Let $\mathcal L_i$ be the longest prefix of $\suff{S}{i}$ that is Lyndon. Then, $\mathcal L_i\lexlt\mathcal L_j$ implies $\suff{S}{i}\lexlt\suff{S}{j}$.
	\label{lemma:lyndon_prefix}
\end{lemma}
\begin{lemma}
	Consider an occurrence of $X_i\to X_aX_b$ at position $j$ in the text $S$ (this implies $S\intervCC{j}{j+\abs{\val{X_i}}} = \val{X_i}$).
	Then, $\val{X_b}$ is the longest prefix of $\suff{S}{j + \abs{\val{X_a}}}$ that is Lyndon.
	\label{lemma:lyndon_prefix_order}
\end{lemma}
\begin{proof}
	Follows immediately from the relationship between Lyndon forest and Lyndon grammar \cite{lyndon_grammar} and \cite[Lemma 15]{franek2016algorithms}.
\end{proof}
\begin{corollary}
	\label{corollary:lyndon_prefix_order}
	Consider occurrences of $X_i\to X_aX_b$ at position $j$ and of $X_{i'}\to X_{a'}X_{b'}$ at position $j'$.
	It holds $\suff{S}{j + \abs{\val{X_a}}} \lexlt \suff{S}{j' + \abs{\val{X_{a'}}}}$ if $\val{X_b} \lexlt \val{X_{b'}}$.
\end{corollary}
To see why Corollary~\ref{corollary:lyndon_prefix_order} is useful, consider a symbol $X_i\to X_aX_b$.
Each occurrence of $X_i$ in the derivation tree corresponds to a suffix of the text with prefix $\val{X_b}$,
and each such suffix introduces an occurrence of the last character of $\val{X_a}$ to the $\BWT$ in the $\SA$-interval of $\val{X_b}$.
In our running example (cf.\ Figure~\ref{fig:tree_slp_example}), each occurrence of $X_8\to X_7X_9$ introduces a $\texttt{b}$ in the $\SA$-interval of $\val{X_9}=\texttt{c}$.


\begin{algorithm}
	\caption{Computing $\SA'_\circ$ from the lexicographically sorted Lyndon grammar \cite[Algorithms~2 and 3]{olbrich2024generic}.}
	\label{alg:sa_from_grammar}
	$(s,\SA'_\circ,\mathsf L) \gets (n, [], []^n)$\;
	\For(\tcc*[h]{insert the end positions of the roots $r_1,\dots,r_k$ of the Lyndon SLP}){$i=k\to 1$}{
		$\mathsf L[2(r_i - 1) + 1].\mathit{append}(s)$\;
		$s\gets s-\abs{\val{X_{r_i}}}$\;
	}
	\For{$m=0\to 2g - 1$}{
		\ForEach{$i$ in $\mathsf L[m]$}{
			$\SA'_\circ.\mathit{append}(i)$\;
			\For{$j$ with $\nss[j] = i$}{
				Find $X_k$ such that $\val{X_k}=S\intervCO{j}{\nss[j]}$\;
				$\mathsf L[2(k-1)].\mathit{append}(j)$\;
			}
		}
	}
\end{algorithm}
In \cite{olbrich2024generic}, an array $\SA'_\circ$\footnote{$\SA'_\circ$ is essentially the analogue to $\SA$ for conjugates of Lyndon factors and the infinite periodic order. A precise definition is omitted because only the relationship between $SA'_\circ$ and $\BBWT$ is relevant for this paper.} is computed such that $\BBWT[i] = S[\SA'_\circ[i]]$ holds for all $i$.
This array arises from sorting the conjugates of the Lyndon factors of $S$ and replacing each start position of a Lyndon word with the end position of the Lyndon word in $S$.
It was shown that Algorithm~\ref{alg:sa_from_grammar} correctly computes the array $\SA'_\circ$ in linear time \cite{olbrich2024generic}.\footnote{In \cite{olbrich2024generic}, Lyndon factors (roots of the Lyndon forest) come after every non-root representing the same Lyndon word. To reflect this, we use $\mathsf L[2(i - 1) + 1]$ if an occurrence of $X_i$ is a root and $\mathsf L[2(i-1)]$ otherwise.}
Basically, Lemma~\ref{lemma:lyndon_prefix} implies that the positions in $\SA'_\circ$ are grouped by the longest Lyndon prefixes of the corresponding suffixes.
The longest Lyndon prefix of $\suff{S}{i}$ is $S\intervCO{i}{\nss[i]}$.
Because $\suff{S}{\nss[i]}\lexlt\suff{S}{i}$ by definition, we can proceed by induction from lexicographically small to large.\footnote{There is a slight detail omitted here for clarity and brevity. Namely, that for the $\BBWT$ we are concerned with the \emph{next smaller conjugate} of the respective Lyndon factor instead of the next smaller suffix \cite{olbrich2024generic}.}
Consider for instance the suffixes starting with $\texttt{c}$ at indices $5$ and $7$ in our running example (cf.\ Figure~\ref{fig:tree_slp_example}).
By Lemma~\ref{lemma:lyndon_prefix}, they are in the same ``Lyndon group'' \cite{baier2016linear} in the sorted list of suffixes because $\texttt{c}$ is their longest Lyndon prefix.
We have $\nss[5] = 5 + \abs{\texttt{c}} = 6$ and $\nss[7] = 7 + \abs{\texttt{c}} = 8$.
When processing these Lyndon groups in lexicographically increasing order, the relative order of suffixes $6$ and $8$ is known at the point in time when the Lyndon group \texttt{c} is considered (because they are lexicographically smaller), and can therefore determine the lexicographical order of the suffixes $5$ and $7$.

We are now going to transform Algorithm~\ref{alg:sa_from_grammar} such that it outputs the $\BBWT$ instead of $\SA'_\circ$ and works without $\nss$ and the text indices.
First, consider a text index $i\in\intervCO0N$ and let $u_i$ be the highest node in the Lyndon forest at position $i$.
(Equivalently, $u_i$ corresponds to the longest Lyndon word starting at $i$.)
Now consider the parent $v_i$ of $u_i$.
By its definition, it is clear that $u_i$ is the right child of $v_i$. Let $w_i$ be the left child of $v_i$.
We can observe that, starting from $w_i$, the nodes on the rightmost path (excluding $w_i$) correspond exactly to the set $\set{j\in\intervCO0i \mid \nss[j]=i}$.
For instance, in our running example there is a node labelled with $\val{X_8}=\texttt{bc}$ at position $6$ (cf.\ Figures~\ref{fig:intro_examples} and \ref{fig:tree_slp_example}).
This node is the right child of a node labelled with $\val{X_6}$, which has a node labelled with $\val{X_5}$ as left child.
There are two nodes on the rightmost path from this latter node, namely one at position $4$ labelled with $\val{X_8}=\texttt{bc}$ and one at position $5$ labelled with $\val{X_9}=\texttt{c}$. We consequently have $\nss[4]=\nss[5] = 6$.

Second, for $\SA'_\circ[k] = i$ we have $\BBWT[k] = S[i-1]$ and it follows that $\BBWT[k]$ is the last character of the string corresponding to $v_i$'s left child $w_i$.
Because of these two observations, we can replace the insertion of $i$ in Algorithm~\ref{alg:sa_from_grammar} with the insertion of $w_i$.
Consequently, we can reformulate Algorithm~\ref{alg:sa_from_grammar} to Algorithm~\ref{alg:bwt_from_grammar}.

\begin{theorem}
	Algorithm~\ref{alg:bwt_from_grammar} correctly computes $\BBWT$ from the sorted Lyndon SLP in $\bigO{N}$ time.
\end{theorem}
\begin{proof}
	Correctness follows from the above observations.
	The linear worst-case time complexity trivially follows from the fact that each iteration of the inner loop appends at least one character to $\BBWT$.
\end{proof}

Note that runs in a list in $\mathsf L$ compound: If there is a run of symbol $X_i$, there is also a run of at least the same length of $X_\ell$ in $X_r$ for each non-terminal $X_s\in \rightmost\bra{X_i}$ with $X_s\to X_\ell X_r$. Thus, Algorithm~\ref{alg:bwt_from_grammar} often requires much fewer than $N$ iterations.

\begin{algorithm}
	\caption{Deriving the $\BBWT$ from the lexicographically sorted Lyndon grammar}
	\label{alg:bwt_from_grammar}
	$\mathsf L \gets []^{2g}$\tcc*{initialize $\mathsf L[i]$ as an empty RLE list $\forall i\in\intervCO{0}{2g}$}
	\lFor(\tcc*[h]{insert the roots $r_1,\dots,r_k$}){$i=1\to k$}{$\mathsf{A}[2(r_i - 1) +1].\mathit{append}(r_i)$}
	\For{$i=0\to 2g - 1$}{
		\For{$(s,\mathit{count})$ \textnormal{\textbf{in}} $\mathsf{L}[i]$}{
			$\BBWT.\mathit{append}(\mathit{last\_char}[s], \mathit{count})$\;
			\While(\tcc*[h]{walk the rightmost path from $X_s$}){$X_s$ is a non-terminal}{
				Let $X_s\to X_a X_b$\;
				$\mathsf{A}[2(b-1)].\mathit{append}(a, \mathit{count})$\;
				$s \gets b$\;
			}
		}
	}
\end{algorithm}

\subsection{Computing $\BWT$ and $\eBWT$}
\label{sec:bbwt_ebwt}

When the text $S$ is a Lyndon word, the $\BBWT$ of $S$ is equal to the original $\BWT$ of $S$ \cite{linearbbwt}.
Since the $\BWT$ is independent of the rotation of the input and $\$S$ is a Lyndon word, we have $\BBWT\bra{\$S}=\BWT\bra{S\$}$. The latter is the \$-$\BWT$ commonly computed via the suffix array.
Therefore, we can simply compute the \$-$\BWT$ by prepending a sentinel character to the text. Note that this is trivial and can be done after building the Lyndon grammar of $S$.

The $\eBWT$ of a set $\mathcal S$ of strings is the same as the $\BBWT$ of the string $\mathcal S^c$ that arises from arranging the canonical forms of the strings in $\mathcal S$ in lexicographically decreasing order \cite{olbrich2024generic}.
In particular, the Lyndon factors of $\mathcal S^c$ are exactly the canonical forms of the input strings (assuming that the input strings are primitive).
Note that, in Algorithm~\ref{alg:bwt_from_grammar}, sorting the canonical forms of the input strings is done implicitly in the first for-loop.

Note that the Lyndon grammars of the input strings are independent of each other (as long as equal Lyndon words are assigned equal grammar symbols).
Therefore, we can also construct their Lyndon grammars independently of each other while using the same dictionary.
As a consequence, it is easy to use our algorithms for parallel construction of the Lyndon grammar of a collection of sequences. 
Although our algorithm for deriving the $\BWT$ from the grammar is not parallelized, parallel construction of the Lyndon grammar leads to a substantial reduction in wall-clock-time because it is by far the most time-consuming part of the pipeline (see Section~\ref{sec:experiments}).

\subsection{Other $\BWT$ variants for string collections}

Besides the original $\eBWT$, several other $\BWT$ variants for string collections have been proposed (see \cite{cenzato2024survey} for an overview).
Several of the variants are especially suited to be computed with our algorithm, because there the input sequences $\mathcal S=\set{S_1,\dots,S_n}$ can be parsed independently (and in parallel) like for the $\eBWT$:
The \emph{dollar-$\eBWT$} $\dolEBWT\bra{\mathcal S}=\eBWT\bra{\set{S\$ \mid S\in\mathcal S}}$,
the \emph{multidollar $\BWT$} $\mdolBWT\bra{\mathcal S} = \BWT(S_1\$_1\dots S_n\$_n)$, and
the \emph{concatenated $\BWT$} $\concBWT\bra{\mathcal S}=\BWT(S_1\$\dots S_n\$\#)$,
where $\#<\$$ and $\$_1<\dots<\$_n$ are smaller than all characters in $\mathcal S$.
This is done by constructing the Lyndon SLP of each $S_i$ (with a shared dictionary) and then applying post-processing steps such that the desired $\eBWT$ variant can be derived from the resulting SLP.
This is possible because the $\$$'s separate the input strings in the sense that no Lyndon word starting inside some $S_i$ can contain a $\$$.

More specifically,
for the $\dolEBWT$, no post-processing is required and we can just apply our $\eBWT$ construction algorithm, except that finding the canonical forms of the strings is trivial because $\eBWT(\{S\$\mid S\in\mathcal S\})=\eBWT(\{\$S\mid S\in\mathcal S\})$ and $\$S$ is a Lyndon word.
For the $\mdolBWT$ and $\concBWT$, we replace each $S_i$ in $S_1\$_1\dots S_n\$_n$ and $S_1\$\dots\$S_n\$\#$, respectively, with the roots of the grammar generating $S_i$ and compute the grammar of a rotation of the resulting string using Algorithm~\ref{alg:bwt_from_grammar}:
both $\$_1S_2\$_2\dots S_n\$_nS_1$ and $\#S_1\$\dots\$S_n\$$ are Lyndon.
Therefore, $\BWT(S_1\$_1\dots S_n\$_n) = \BBWT(\$_1S_2\$_2\dots S_n\$_nS_1)$ and $\BWT(S_1\$\dots\$S_n\$\#) = \BBWT(\#S_1\$\dots\$S_n\$)$
and we can apply our $\BBWT$ construction algorithm to these rotations.

Note that all other $\BWT$ variants for string collections given in \cite{cenzato2024survey} can be simulated with the multidollar $\BWT$ by using different relative orders of the separator symbols.

\section{Practical Lyndon grammar construction}
\label{sec:hash_construction}

As shown in \cite{badkobeh2022back}, the simple folklore algorithm for constructing the Lyndon array
can be used to construct the Lyndon forest.
With a fitting lookup data structure, we can trivially construct the Lyndon SLP instead.
We essentially use Algorithm~\ref{alg:online_basic}.

\begin{algorithm}
	\caption{Prepending $S[i]$ \cite{badkobeh2022back}}
	\label{alg:online_basic}
	Find $c$ with $\val{X_c} = S[i]$\Comment*[r]{Set $c$ to the symbol generating $S[i]$}
	\While{$\mathsf{stack}$ is not empty}{
		$t \gets \mathsf{stack}.\mathsf{top()}$\;
		\lIf{$\val{X_t} \lexleq \val{X_c}$}{\Break}
		\Else{
			$\mathsf{stack}.\mathsf{pop()}$\;
			Find $X_{c'}$ with $X_{c'}\to X_cX_t$\;
			$c \gets c'$\;
		}
	}
	$\mathsf{stack}.\mathsf{push(}\mathsf{c)}$\;
\end{algorithm}

There are two steps that have to be explained:
First, one must be able to compare $\val{X_t}$ and $\val{X_c}$ lexicographically, and secondly, one must be able to find $X_{c'}$ with $X_{c'}\to X_cX_t$ (this is also called the \emph{naming function}).
For the former, we store for each symbol on the stack a fixed-length prefix of the generated string. In many cases, this is sufficient for determining the lexicographical order. Possibilities to handle the other cases are explained in Sections~\ref{sec:pss_comp}, \ref{sec:naive_comp} and \ref{sec:btree_construction}.
The following paragraphs will deal with the naming function.

For the naming function, we use a hash table to find the symbol names. This provides (expected) constant time lookup/insertion.\footnote{
Note that this is also possible in $\bigO{\log\log g}$ deterministic time with at most $2g + g\log g + o\bra{g\log g})$ bits of memory, where $g$ is the number of symbols of the grammar \cite{takabatake_et_al:LIPIcs.ESA.2017.67}.
}
In practice however, $N$ hash table lookups are unnecessarily slow, especially because we assume the input to be repetitive.
In particular, having fewer but longer keys is cache friendlier and thus faster on modern computer architectures.

Specifically, we decrease the number of hash table accesses by partitioning the nodes of the Lyndon forest into \emph{heavy} and \emph{light} nodes and querying a second hash table $H$ for each heavy node. The naming function is then only used when such a query to $H$ does not return an answer.
Leaves of the Lyndon forest are considered heavy. An inner node $v$ is considered heavy, if and only if its number of \emph{immediate heavy descendants} (heavy descendants for which only light nodes occur on the path to $v$) exceeds a fixed constant $n_{\mathit{thres}} > 1$ (see Figure~\ref{fig:heavy_light}).\footnote{In our experiments, $n_{\mathit{thres}} = 31$.}

\begin{figure}
	\centering
	\def\examplespacing{0.5}
	\def\treevspacing{0.4cm}
	\newcommand{\treenode}[5]{
		\draw node[inner sep=2pt,circle,fill=#5] (#1) at (#2,#3) {};
	}
	\newcommand{\treen}[5]{
		\draw let \p1 = (#2), \p2 = (#3) in node[inner sep=2pt,circle,fill=#5,draw=#4] (#1) at (\x1*0.5+\x2*0.5,{max(\y1,\y2)+\treevspacing}) {};
		\draw[gray] (#1) -- (#2);
		\draw[gray] (#1) -- (#3)
	}
	\colorlet{lightnode}{gray}
	\begin{tikzpicture}
		\begin{scope}
			\treenode{n00}{0*\examplespacing}{0}{factor_col}{factor_col};
			\treenode{n10}{1*\examplespacing}{0}{factor_col}{factor_col};
			\treenode{n20}{2*\examplespacing}{0}{factor_col}{factor_col};
			\treenode{n30}{3*\examplespacing}{0}{factor_col}{factor_col};
			\treenode{n40}{4*\examplespacing}{0}{factor_col}{factor_col};
			\treenode{n50}{5*\examplespacing}{0}{factor_col}{factor_col};
			\treenode{n60}{6*\examplespacing}{0}{factor_col}{factor_col};
			\treenode{n70}{7*\examplespacing}{0}{factor_col}{factor_col};
			\treen{n01}{n00}{n10}{lightnode}{white};
			\treen{n11}{n20}{n30}{lightnode}{white};
			\treen{n21}{n50}{n60}{lightnode}{white};
			\treen{n12}{n40}{n21}{factor_col}{factor_col};
			\treen{n02}{n01}{n11}{factor_col}{factor_col};
			\treen{n03}{n02}{n12}{lightnode}{white};
			\treen{n04}{n03}{n70}{factor_col}{factor_col};
		\end{scope}
	\end{tikzpicture}
	\caption{
		Partial representation of a Lyndon Forest with $n_\mathit{thres} = 2$.
		Heavy nodes are coloured~in~({\color{factor_col}$\blacksquare$}) while light nodes are not.
	}
	\label{fig:heavy_light}
\end{figure}

For a heavy node $v$ of the Lyndon forest we query $H$ with the sequence of immediate heavy descendants $v_1,\dots,v_k$ of $v$.
This means that we only store heavy nodes on the stack, for a light node we instead store all immediate heavy descendants.

Different occurrences of the same symbol represent the exact same Lyndon tree and therefore correspond to the exact same sequence of immediate heavy descendants.
Consequently, if a subtree isomorphic to the subtree rooted at a heavy node has occurred previously, the query to $H$ is sufficient to resolve it.
Otherwise, we use the naming function to find the correct symbol and insert the sequence of immediate heavy descendants with it into $H$.

Note that each heavy node has at most $2n_{\mathit{thres}}$ immediate heavy descendants because it has exactly two children (or zero), each of which is either light (and thus has at most $n_{\mathit{thres}}$ immediate heavy descendants) or heavy itself.
This implies that the total memory usage increases by a factor of at most $\bigO{n_{\mathit{thres}}}$.

\subsection{Constant time suffix comparisons}
\label{sec:pss_comp}

Because the Lyndon forest (and thus the Lyndon SLP) is closely related to the Lyndon array, we can use the Lyndon array to decide whether $\val{X_t}\lexleq \val{X_c}$ holds in constant time.
Specifically, when considering $S[i]$ in Algorithm~\ref{alg:bwt_from_grammar}, we have $\val{X_c}\lexlt\val{X_t}$ while $\lambda[i] > \abs{\val{X_c}}$.

Unfortunately, as far as we know, all linear-time algorithms for constructing the Lyndon array directly from the text need random access to the text \cite{badkobeh2022back,bille2019space,ellert:LIPIcs.ESA.2022.48} and thus the text should reside in main memory.

For a $\BWT$ of a real-world string collection $\mathcal S$ however, this is typically not a problem because, as noted in Section~\ref{sec:bbwt_ebwt}, the SLPs of the strings in $\mathcal S$ can be computed independently (with a shared dictionary).
Consequently, it suffices to have only the string (and its (succinct) Lyndon array) in main memory whose Lyndon SLP is currently computed.
This increases the RAM usage by $n_{\mathrm{max}}(\lceil\log_2\abs\Sigma\rceil + 2)$ bits, where $n_{\mathrm{max}}$ is the length of the longest string in $\mathcal S$.
For real biological data, $n_{\mathrm{max}}$ is generally small compared to the size of the entire input.

\subsection{Na{\"i}ve suffix comparisons using the Lyndon grammar}
\label{sec:naive_comp}

In this section, we describe a simple method for comparing $\val{X_t}$ and $\val{X_c}$ lexicographically which requires constant extra memory and works well in practice.
We assume that $a,b < i$ holds for each rule $X_i\to X_aX_b$. This is clearly the case for SLPs constructed with Algorithm~\ref{alg:online_basic} if we assign increasing indices to new symbols.

First note that, because we assign equal Lyndon words to equal symbols, $\val{X_t} = \val{X_c}$ if and only if $t = c$.
The key to our algorithm is that we find the longest common symbol $X_\ell\in\leftmost\bra{X_t}\cap\leftmost\bra{X_c}$.
Additionally, we determine $X_{t'}\in\leftmost\bra{X_t}$ and $X_{c'}\in\leftmost\bra{X_c}$ with $X_{t'}\to X_\ell X_{r}$ and $X_{c'}\to X_\ell X_{r'}$.
Note that, by definition, $r\neq r'$ (otherwise, $c'=t'$, contradicting the choice of $\ell$) and therefore the lexicographical order of $\val{X_r}$ and $\val{X_{r'}}$ is the same as the lexicographical order of $\val{X_t}$ and $\val{X_c}$.
Formally, $\val{X_t}\lexlt\val{X_c} \iff \val{X_r}\lexlt\val{X_{r'}}$.
If such a tuple $(\ell,t',c')$ does not exist, either $\val{X_t}$ and $\val{X_c}$ do not share a non-empty prefix (i.e., $\val{X_t}[0]\neq\val{X_c}[0]$),
or one of $X_t$ and $X_c$ is a prefix of the other.

Note that the indices of the elements in $\leftmost\bra{X_i}$ have the same relative order as the lengths of the generated strings in the sense that
for all $a,b\in\leftmost\bra{X_i}$ we have $a<b$ if and only if $\abs{\val{X_a}}<\abs{\val{X_b}}$.
For this reason, we can proceed with a two-pointer search to find $\ell$ (and $t'$ and $c'$).
More specifically, assuming that the desired $\ell\in\leftmost\bra{X_t}\cap\leftmost\bra{X_c}$ exists, we have $\ell\in\leftmost\bra{X_a}$ if $\abs{\val{X_t}} < \abs{\val{X_c}}$ for $X_c\to X_aX_b$, and vice versa.
Algorithm~\ref{alg:naive_comp} shows the procedure.
\begin{algorithm}
	\caption{Comparing two symbols lexicographically}
	\label{alg:naive_comp}
	\SetKwFunction{compare}{compare}
	\SetKwProg{myproc}{Procedure}{}{}
	\myproc{\compare{$X_t$, $X_c$}}{
		\lIf{$t=c$}{\Return $\val{X_t}=\val{X_c}$}
		$((\ell,\ell'), (r,r')) \gets ((t,\bot),(c,\bot))$\;
		\While{$X_\ell$ and $X_r$ are non-terminals with $X_\ell\to X_aX_b$, $X_r\to X_{a'}X_{b'}$}{
			\lIf(\tcc*[f]{note that $b\neq b'$}){$a=a'$}{\Return \compare{$X_b$, $X_{b'}$}}
			\lIf{$a<a'$}{$(r,r')\gets(a',b')$}
			\lElse{$(\ell,\ell')\gets(a,b)$}
		}
		\lWhile{$X_\ell$ is a non-terminal with $X_\ell\to X_aX_b$}{ $(\ell,\ell')\gets (a,b)$ }
		\lWhile{$X_r$ is a non-terminal with $X_r\to X_{a'}X_{b'}$}{ $(r,r')\gets (a',b')$ }
		\tcc{$X_\ell$ and $X_r$ are both terminal symbols. Assume $X_\ell\to x$, $X_r\to x'$}
		\If{$x \neq x'$}{
			\lIf{$x > x'$}{\Return $\val{X_t}\lexlt\val{X_c}$}
			\lElse{\Return $\val{X_t}\lexgt\val{X_c}$}
		}
		\tcc{$l=r$ and $\ell'\neq r'$}
		\lIf{$\ell'\neq\bot\land r'\neq\bot$}{\Return \compare{$X_{\ell'}$,$X_{r'}$}}
		\lElseIf(\tcc*[f]{$\val{X_t}\text{ is a prefix of }\val{X_c}$}){$\ell'=\bot$}{\Return $\val{X_t}\lexleq\val{X_c}$}
		\lElse(\tcc*[f]{$r'=\bot\implies\val{X_c}\text{ is a prefix of }\val{X_t}$}){\Return $\val{X_t}\lexgt\val{X_c}$}
	}
\end{algorithm}

Because in each step, at least one of the symbol indices decreases, the time complexity for a comparison is at most linear in the size of the SLP.
In fact, because in each step we go from a symbol to one of its children, the time complexity is actually bounded by the \emph{height} of the SLP.
This in turn implies a time complexity of $\bigO{Ng}$ when using Algorithm~\ref{alg:online_basic} with the described method for constructing the Lyndon SLP.

Note that this worst-case time complexity is tight, e.g.\ on the string $a^kba^k$ ($k\in\mathbb N$).
However, the Lyndon forests of random strings have expected height proportional to the logarithm of the input size \cite{mercier2013height} and our approach works very well in practice (see Section~\ref{sec:experiments}).

\subsection{Construction in $\bigO{N\log g + g\log^2g}$}
\label{sec:btree_construction}

In this Section, we describe an algorithm that is able to compute the Lyndon grammar online in $\bigO{N\log g + g\log^2g}$ deterministic time
from right to left in a streaming fashion using $\bigO{g}$ words of extra memory.

Basically, this is done by maintaining the grammar's set of symbols $\mathcal S$ in an ordered sequence, lexicographically sorted by their respective generated strings.
Using e.g.\ a B-Tree \cite{bayer1970organization}, one can then find the rank of any symbol in $\bigO{\log\abs{\mathcal S}}$ time.
Thus, determining the lexicographical order of two symbols in $\mathcal S$ can also be done in $\bigO{\log\abs{\mathcal S}}$ time.
What remains to be shown is how the symbols can be maintained in this sorted arrangement.

As shown in Section~\ref{sec:lyndon_grammar_properties}, a lexicographically sorted Lyndon grammar has a forest-structure (cf.\ Figure~\ref{fig:interval_forest}).
This first-symbol forest can be represented using a \emph{balanced parenthesis sequence (BPS)} of length $2\abs{\mathcal S}$ \cite{succinctTrees}, which can be obtained using a depth-first traversal of the first-symbol forest (starting at the roots) by writing an opening parenthesis \texttt{`('} when visiting a node for the first time and a closing parenthesis \texttt{`)'} when all subtrees of a node have been visited \cite{succinctTrees}.
Each parenthesis pair corresponds to a symbol in the grammar, where the $i$th opening parenthesis corresponds to the $i$th smallest symbol (by lexicographical order).
For a grammar with symbol set $\mathcal S$, let this BPS be $\mathcal B_{\mathcal S}$.
We represent $\mathcal B_{\mathcal S}$ as an ordered sequence $\mathcal T_{\mathcal S}$, which contains two markers $\texttt(_i$ and $\texttt)_i$ for each symbol $X_i$ in $\mathcal S$, such that the ranks of $\texttt(_i$ and $\texttt)_i$ are the indices of $X_i$'s opening and closing parenthesis in $\mathcal B_{\mathcal S}$, respectively.
For instance, the sequence $\mathcal T_{\mathcal S}$ for the lexicographically sorted Lyndon grammar of our running example (see Figure~\ref{fig:interval_forest}) would be $\texttt(_1\texttt(_2\texttt(_3\texttt(_4\texttt)_4\texttt)_3\texttt)_2\texttt(_5\texttt(_6\texttt)_6\texttt)_5\texttt)_1\texttt(_7\texttt(_8\texttt)_8\texttt)_7\texttt(_9\texttt)_9$.

Now consider adding a new symbol $X_i\to X_aX_b$, where $X_a$ and $X_b$ are in $\mathcal T$ (i.e.\ $X_i\notin\mathcal S$, $X_a,X_b\in\mathcal S$).
By Lemma~\ref{cor:symbol_tree}, the parent of $X_i$'s parenthesis pair in $\mathcal B_{\mathcal S\cup\set{X_i}}$ must be $X_a$'s parenthesis pair.
Therefore, it suffices to determine $X_i$'s lexicographically smallest ``sibling'' $X_j\to X_aX_c$ ($X_j\in\mathcal S$) with $\val{X_c} \lexgt \val{X_b}$; $X_i$'s parenthesis pair must appear immediately in front of $\texttt(_j$ in $\mathcal T_{\mathcal S\cup\set{X_i}}$.
If there is no such sibling, $X_i$ is (currently) the largest child of $X_a$ and thus $X_i$'s parenthesis pair must be immediately in front of $\texttt)_a$ instead.
Note that comparing $\val{X_b}$ and $\val{X_c}$ is possible in $\bigO{\log\abs{\mathcal{S}}}$ because both $X_b$ and $X_c$ are in $\mathcal T_{\mathcal S}$.

In order to be able to find the correct sibling of $X_i$ as described, we additionally maintain for each $X_a$ an ordered sequence $\mathcal T_a$ containing the symbols $\mathcal S_a = \set{X_k\to X_aX_b\mid X_k,X_b\in\mathcal S}$ in lexicographical order.
Inserting $X_i\to X_aX_b$ into $\mathcal T_a$ can then be accomplished with $\bigO{\log\abs{\mathcal S_a}}\subseteq\bigO{\log\abs{\mathcal S}}$ comparisons, each of which is possible in $\bigO{\log{\abs{\mathcal S}}}$ via $\mathcal T_{\mathcal S}$.

In total, we obtain a time complexity of $\bigO{\abs{\mathcal S}\log^2\abs{\mathcal S}}$ for maintaining the grammar's symbols lexicographically sorted, and $\bigO{N\log \abs{\mathcal S}}$ for constructing the Lyndon forest.

\section{Experiments}
\label{sec:experiments}

\pgfplotscreateplotcyclelist{colorlist}{%
mark=x,red\\%
mark=x,black\\%
mark=+,teal,dashed\\%
mark=+,magenta,dashed\\%
mark=+,blue,dashed\\%
mark=+,green!60!black,dashed\\%
mark=+,brown,densely dashed\\
mark=+,gray,densely dotted\\
}
\pgfplotsset{
	width=7cm,height=5.2cm,
	cycle list name=colorlist,
	every axis plot/.append style={every mark/.append style={solid,line width=0.2mm},mark size=2pt,line width=0.1mm},
}

\begin{figure}[t]
	\small
	\begin{tikzpicture}
		\begin{axis}[
			xmode=log,ymode=log,
			xlabel=number of haplotypes,
			ylabel=time (s),
			legend style={anchor=south,at={(11.5cm,6.5cm)},legend columns=3,cells={anchor=west}},
			at={(-0.75cm,0cm)},
			anchor=north east]

			\addplot plot coordinates {
				(50,105.220)
				(100,179.430)
				(250,400.140)
				(500,755.370)
				(750,1139.760)
				(1000,1492.700)
			};
			\addlegendentry{Ours $\bigO{n\log g + g\log^2g}$}

			\addplot plot coordinates {
				(50,81.280)
				(100,152.870)
				(250,350.710)
				(500,693.720)
				(750,1036.100)
				(1000,1373.100)
			};
			\addlegendentry{Ours $\bigO{ng}$}

			\addplot plot coordinates {
			(50,176.380)
			(100,272.4)
			(250,556.0)
			(500,1032.0)
			(750,1524.1)
			(1000,1982.7)
			};
			\addlegendentry{\bigbwt{}\textsubscript{$500$}}

			\addplot plot coordinates {
			(50,150.3)
			(100,244.0)
			(250,536.2)
			(500,1051.5)
			(750,1524.1)
			(1000,1982.7)
			};
			\addlegendentry{\bigbwt{}\textsubscript{$200$}}

			\addplot plot coordinates {
			(50,169.0)
			(100,289.2)
			(250,665.4)
			(500,1314.4)
			(750,2007.0)
			(1000,2626.8)
			};
			\addlegendentry{\bigbwt{}\textsubscript{$100$}}

			\addplot plot coordinates {
			(50,225.6)
			(100,411.7)
			(250,992.1)
			(500,1961.0)
			(750,2989.6)
			(1000,4045.0)
			};
			\addlegendentry{\bigbwt{}\textsubscript{$50$}}

			\addplot plot coordinates {
			(50,178.0)
			(100,247.8)
			(250,474.2)
			(500,918.1)
			(750,1336.1)
			(1000,1758.4)
			};
			\addlegendentry{\rpfbwt{}}

			\addplot plot coordinates {
			(50,177.3)
			(100,355.8)
			(250,1082.6)
			(500,2610.2)
			(750,4195.0)
			(1000,5934.0)
			};
			\addlegendentry{\libsais{}}
		\end{axis}
		\begin{axis}[
			xmode=log,ymode=log,
			xlabel=number of haplotypes,
			ylabel=RAM (GiB),
			at={(0.75cm,0cm)},
			anchor=north west]

			\addplot plot coordinates {
				(50,2.133)
				(100,2.132)
				(250,2.314)
				(500,3.287)
				(750,4.116)
				(1000,4.614)
			};

			\addplot plot coordinates {
			(50,0.903)
			(100,0.903)
			(250,1.027)
			(500,1.542)
			(750,1.553)
			(1000,1.747)
			};

			\addplot plot coordinates {
			(50,2.40)
			(100,2.99)
			(250,4.51)
			(500,7.44)
			(750,10.37)
			(1000,13.47)
			};

			\addplot plot coordinates {
			(50,1.44)
			(100,1.69)
			(250,2.39)
			(500,3.80)
			(750,5.14)
			(1000,6.83)
			};

			\addplot plot coordinates {
			(50,1.07)
			(100,1.28)
			(250,1.87)
			(500,2.92)
			(750,3.93)
			(1000,5.48)
			};

			\addplot plot coordinates {
			(50,1.07)
			(100,1.37)
			(250,2.41)
			(500,4.81)
			(750,7.21)
			(1000,9.62)
			};

			\addplot plot coordinates {
			(50,1.94)
			(100,2.22)
			(250,3.05)
			(500,4.65)
			(750,6.25)
			(1000,8.37)
			};

			\addplot plot coordinates {
			(50,24.78)
			(100,49.56)
			(250,123.9)
			(500,247.79)
			(750,371.69)
			(1000,495.58)
			};
		\end{axis}
	\end{tikzpicture}
	\caption{
		Wall clock time and maximum resident memory of $\BWT$ construction algorithms on chromosome 19 sequences.
		The sequences were concatenated and all programs used only one thread.
	}
	\label{fig:bwt:chr19}
\end{figure}

\pgfplotscreateplotcyclelist{colorliste}{%
mark=x,black\\%
mark=x,orange\\
mark=x,gray\\
mark=x,yellow\\
mark=+,brown,densely dashed\\%
mark=+,gray,densely dotted\\
mark=+,magenta,densely dashed\\%
mark=+,blue,densely dashed\\%
mark=+,green!60!black,densely dashed\\%
}
\pgfplotsset{
	cycle list name=colorliste,
}
\begin{figure}
	\small
	\begin{tikzpicture}
		\begin{axis}[
			xmode=log,ymode=log,
			xlabel=number of haplotypes,
			ylabel=time (s),
			legend style={anchor=south,at={(11.5cm,5.25cm)},legend columns=3,cells={anchor=west}},
			at={(-0.75cm,0cm)},
			anchor=north east]

			\addplot plot coordinates {
				(50,87.1)
				(100,158.5)
				(250,385.5)
				(500,733.8)
				(750,1092.8)
				(1000,1443.9)
			};
			\addlegendentry{Ours $\$\eBWT$ $\bigO{Ng}$}

			\addplot plot coordinates {
				(50,130.8)
				(100,246.8)
				(250,590.9)
				(500,1170.6)
				(750,1766.8)
				(1000,2324.6)
			};
			\addlegendentry{Ours $\$\eBWT$ $\bigO{N}$}

			\addplot plot coordinates {
				(50,91.890)
				(100,170.130)
				(250,403.860)
				(500,793.020)
				(750,1181.810)
				(1000,1574.440)
			};
			\addlegendentry{Ours $\eBWT$ $\bigO{Ng}$}

			\addplot plot coordinates {
				(50,136.140)
				(100,257.870)
				(250,621.170)
				(500,1235.900)
				(750,1842.890)
				(1000,2444.770)
			};
			\addlegendentry{Ours $\eBWT$ $\bigO{N}$}

			\addplot plot coordinates {
				(50,155.330)
				(100,212.510)
				(250,383.880)
				(500,716.370)
				(750,1039.340)
				(1000,1379.700)
			};
			\addlegendentry{\rpfbwt}

			\addplot plot coordinates {
				(50,233.140)
				(100,499.900)
				(250,6429.000)
				(500,18153.000)
				(750,30796.000)
				(1000,42907.000)
			};
			\addlegendentry{\ropebwt}

			\addplot plot coordinates {
				(50,235.270)
				(100,456.860)
				(250,1125.050)
				(500,2271.800)
				(750,3380.740)
				(1000,4499.000)
			};
			\addlegendentry{\grlbwt}

			\addplot plot coordinates {
				(50,140.700)
				(100,258.960)
				(250,629.170)
				(500,1261.350)
				(750,1877.380)
				(1000,2541.610)
			};
			\addlegendentry{\pfpebwt{}\textsubscript{$100$}}

			\addplot plot coordinates {
				(50,58.310)
				(100,99.950)
				(250,233.230)
				(500,471.160)
				(750,733.420)
				(1000,982.760)
			};
			\addlegendentry{\cmsbwt}
		\end{axis}
		\begin{axis}[
			xmode=log,ymode=log,
			xlabel=number of haplotypes,
			ylabel=RAM (GiB),
			at={(0.75cm,0cm)},
			anchor=north west]

			\addplot plot coordinates {
				(50,0.959)
				(100,0.960)
				(250,1.086)
				(500,1.599)
				(750,1.681)
				(1000,2.036)
			};
			
			\addplot plot coordinates {
				(50,0.959)
				(100,0.960)
				(250,1.089)
				(500,1.593)
				(750,1.713)
				(1000,2.072)
			};

			\addplot plot coordinates {
				(50,0.959)
				(100,0.959)
				(250,1.082)
				(500,1.597)
				(750,1.684)
				(1000,2.010)
			};

			\addplot plot coordinates {
				(50,0.957)
				(100,0.957)
				(250,1.088)
				(500,1.592)
				(750,1.695)
				(1000,2.005)
			};

			\addplot plot coordinates {
				(50,1.943)
				(100,2.222)
				(250,3.055)
				(500,4.653)
				(750,6.135)
				(1000,7.391)
			};

			\addplot plot coordinates {
				(50,30.346)
				(100,60.632)
				(250,72.322)
				(500,72.323)
				(750,72.332)
				(1000,72.362)
			};

			\addplot plot coordinates {
				(50,0.149)
				(100,0.227)
				(250,0.355)
				(500,1.076)
				(750,1.117)
				(1000,1.121)
			};

			\addplot plot coordinates {
				(50,1.080)
				(100,1.265)
				(250,1.785)
				(500,3.129)
				(750,4.680)
				(1000,6.232)
			};
			
			\addplot plot coordinates {
				(50,2.413)
				(100,3.144)
				(250,4.173)
				(500,6.481)
				(750,9.089)
				(1000,11.573)
			};
		\end{axis}
	\end{tikzpicture}
	\caption{
		Wall clock time and maximum resident memory of tools for constructing $\BWT$ variants for string collection on chromosome 19 sequences.
		All programs used only one thread.
	}
	\label{fig:ebwt:chr19}
\end{figure}

\pgfplotscreateplotcyclelist{colorliste}{%
mark=x,black\\%
mark=x,orange\\
mark=+,brown,densely dashed\\
mark=+,gray,densely dotted\\%
mark=+,magenta,densely dashed\\%
mark=+,blue,densely dashed\\%
mark=+,green!60!black,densely dashed\\%
}
\pgfplotsset{
	cycle list name=colorliste,
}
\begin{figure}
	\small
	\begin{tikzpicture}
		\begin{axis}[
			ymode=log,xmode=log,
			ylabel=time (s),
			xticklabels={1,4,8,16,32},xtick={1,4,8,16,32},
			legend style={anchor=south,at={(5.5cm,5.9cm)},legend columns=3,cells={anchor=west}},
			at={(-0.75cm,0cm)},
			anchor=north east]

			\addplot plot coordinates {
				(1,1401.880)
				(4,450.550)
				(8,251.720)
				(16,151.550)
				(32,120.240)
			};
			\addlegendentry{Ours $\$\eBWT$ $\bigO{Ng}$}

			\addplot plot coordinates {
				(1,2113.510)
				(4,635.050)
				(8,343.730)
				(16,199.540)
				(32,150.190)
			};
			\addlegendentry{Ours $\$\eBWT$ $\bigO{N}$}

			\addplot plot coordinates {
				(1,1379.700)
				(4,1235.480)
				(8,1211.830)
				(16,1199.340)
				(32,1216.640)
			};
			\addlegendentry{\rpfbwt}

			\addplot plot coordinates {
				(1,42907.000)
				(4,12808.000)
				(8,7870.000)
				(16,5623.000)
				(32,4440.000)
			};
			\addlegendentry{\ropebwt}

			\addplot plot coordinates {
				(1,4499.000)
				(4,1352.990)
				(8,819.820)
				(16,566.640)
				(32,471.880)
			};
			\addlegendentry{\grlbwt}

			\addplot plot coordinates {
				(1,2541.610)
				(4,2211.900)
				(8,3966.000)
				(16,7841.000)
				(32,8085.000)
			};
			\addlegendentry{\pfpebwt{}\textsubscript{$100$}}

			\addplot plot coordinates {
				(1,982.760)
			};
			\addlegendentry{\cmsbwt}
		\end{axis}
		\begin{axis}[
			xmode=log,ymode=log,
			ylabel=RAM (GiB),
			xticklabels={1,4,8,16,32},xtick={1,4,8,16,32},
			at={(0.75cm,0cm)},
			anchor=north west]

			\addplot plot coordinates {
				(1,2.009)
				(4,3.049)
				(8,3.322)
				(16,4.516)
				(32,6.680)
			};
			
			\addplot plot coordinates {
				(1,1.989)
				(4,3.098)
				(8,3.216)
				(16,4.461)
				(32,6.370)
			};

			\addplot plot coordinates {
				(1,7.391)
				(4,7.391)
				(8,7.391)
				(16,7.391)
				(32,7.391)
			};

			\addplot plot coordinates {
				(1,72.362)
				(4,72.553)
				(8,72.746)
				(16,73.136)
				(32,73.911)
			};

			\addplot plot coordinates {
				(1,1.121)
				(4,8.336)
				(8,8.409)
				(16,8.553)
				(32,8.843)
			};

			\addplot plot coordinates {
				(1,6.232)
				(4,6.232)
				(8,6.232)
				(16,6.232)
				(32,6.232)
			};

			\addplot plot coordinates {
				(1,11.573)
			};
		\end{axis}

		\begin{axis}[
			ymode=log,xmode=log,
			xlabel=number of threads,
			ylabel=time (s),
			xticklabels={1,4,8,16,32},xtick={1,4,8,16,32},
			anchor=north east,
			at={(-0.6,-4.5cm)}]

			\addplot plot coordinates {
				(1,663.130)
				(4,183.200)
				(8,99.440)
				(16,59.520)
				(32,40.820)
			};

			\addplot plot coordinates {
				(1,1166.450)
				(4,312.980)
				(8,164.460)
				(16,92.860)
				(32,60.630)
			};

			\addplot plot coordinates {
				(1,2867.680)
				(4,1116.300)
				(8,951.390)
				(16,890.810)
				(32,865.740)
			};

			\addplot plot coordinates {
				(1,18583.000)
				(4,6615.000)
				(8,4648.000)
				(16,3732.000)
				(32,3293.060)
			};

			\addplot plot coordinates {
				(1,1824.860)
				(4,558.040)
				(8,364.830)
				(16,265.240)
				(32,223.670)
			};

			\addplot plot coordinates {
				(1,2082.830)
				(4,1742.820)
				(8,1627.370)
				(16,1625.080)
				(32,3510.680)
			};

			\addplot plot coordinates {
				(1,781.550)
			};
		\end{axis}

		\begin{axis}[
			xmode=log,ymode=log,
			xlabel=number of threads,
			ylabel=RAM (GiB),
			xticklabels={1,4,8,16,32},xtick={1,4,8,16,32},
			anchor=north west,
			at={(0.75cm,-4.5cm)}]

			\addplot plot coordinates {
				(1,1.069)
				(4,2.392)
				(8,2.546)
				(16,2.406)
				(32,2.681)
			};
			
			\addplot plot coordinates {
				(1,1.082)
				(4,2.388)
				(8,2.467)
				(16,2.431)
				(32,2.633)
			};

			\addplot plot coordinates {
				(1,5.433)
				(4,5.433)
				(8,5.433)
				(16,5.433)
				(32,5.433)
			};

			\addplot plot coordinates {
				(1,110.882)
				(4,111.068)
				(8,111.253)
				(16,111.619)
				(32,112.363)
			};

			\addplot plot coordinates {
				(1,0.209)
				(4,4.189)
				(8,4.252)
				(16,4.377)
				(32,4.626)
			};

			\addplot plot coordinates {
				(1,3.433)
				(4,3.433)
				(8,3.433)
				(16,3.433)
				(32,3.433)
			};

			\addplot plot coordinates {
				(1,11.925)
			};
		\end{axis}
	\end{tikzpicture}
	\caption{
		Wall clock time and maximum resident memory of tools for constructing $\BWT$ variants on $1000$ Chromosome~19 sequences (top) and $10^6$ SARS-CoV-2 sequences (bottom) using multiple threads. Note that \cmsbwt does not support multithreading. 
		For the SARS-CoV-2 data, \pfpebwt required the \texttt{----reads} flag.
	}
	\label{fig:ebwt:chr19:par}
\end{figure}

The source code of our implementation is publicly available.\footnote{\url{https://gitlab.com/qwerzuiop/lyndongrammar}}

We compare our $\BWT$ algorithms for single texts with the programs \bigbwt\footnote{\url{https://gitlab.com/manzai/Big-BWT}, last accessed: 22.04.2025, git hash \texttt{944cb27}} \cite{boucher2019prefix} and \rpfbwt\footnote{\url{https://github.com/marco-oliva/r-pfbwt}, last accessed: 22.04.2025, git hash \texttt{1fea5c3}} \cite{oliva2023recursive}, as well as \libsais.\footnote{\url{https://github.com/IlyaGrebnov/libsais}, last accessed: 22.04.2025, git hash \texttt{a138159}}
The latter uses a modified version of the Suffix-Array Induced Sorting (SAIS) algorithm \cite{nong2009linear} to compute the $\BWT$ and, since it is based on the currently fastest $\SA$ construction implementation for general real-world data \cite{olbrich2024generic}, can be viewed as a lower bound for algorithms using the suffix array to compute the $\BWT$.

For the $\BWT$ of text collections, we compare with \pfpebwt\footnote{\url{https://github.com/davidecenzato/PFP-eBWT}, last accessed: 22.04.2025, git hash \texttt{4ca75ce}} \cite{boucher2021computing}, \rpfbwt \cite{oliva2023recursive}, \ropebwt\footnote{\url{https://github.com/lh3/ropebwt3}, last accessed: 22.04.2025, git hash \texttt{36a6411}} \cite{li2024bwt}, \grlbwt\footnote{\url{https://github.com/ddiazdom/grlBWT}, last accessed: 22.04.2025, git hash \texttt{f09e7fa}} \cite{diaz2023efficient} and \cmsbwt\footnote{\url{https://github.com/fmasillo/cms-bwt}, last accessed: 22.04.2025, git hash \texttt{1099d07}. Note that the speed and memory usage of \cmsbwt has improved massively since its publication in \cite{masillo:LIPIcs.ESA.2023.83}.} \cite{masillo:LIPIcs.ESA.2023.83} (for \cmsbwt we used the first sequence in the collection as reference).
All tool except the last one support multi-threading.
Note that not all of these tools compute the same $\BWT$ variant \cite{cenzato2024survey}.
Also note that all algorithms based on PFP as well as \grlbwt are semi-external, i.e., write/read some temporary data to/from disk.

As test data, we use up to 1000 human Chromosome~19 haplotypes from \cite{boucher2021phoni} ($\approx6\cdot10^{10}$bp) and $10^6$ SARS-CoV-2 sequences ($\approx 3\cdot10^{10}$bp).\footnote{Downloaded from \url{https://www.covid19dataportal.org} on 28.08.2024.}
All experiments were conducted on a Linux-6.8.0 machine with an Intel~Xeon~Gold~6338 CPU and 512~GB of RAM.
All programs were compiled with GCC~13.3.0.
Before each test, the test file was scanned once to ensure it is cached by the kernel.
The results can be seen in Figures~\ref{fig:bwt:chr19}, \ref{fig:ebwt:chr19} and \ref{fig:ebwt:chr19:par}.
The subscripts for the PFP-based algorithms indicate the used modulus.

For our programs, computing the $\dolEBWT$ using the suffix comparison method from Section~\ref{sec:naive_comp} is generally the fastest. Computing the original $\eBWT$ is slightly slower because we need to find the smallest rotation of each input string.
In the single-threaded case, for both the 1000 Chromosome~19 haplotypes and $10^6$ SARS-CoV-2 sequences, over $98\%$ of the time was spent constructing the grammar (for our fastest algorithm).
Using the linear-time algorithm for constructing the Lyndon array from \cite{bille2019space} to ensure expected linear time complexity slows our programs down by up to $60\%$.
As expected, the suffix comparison method from Section~\ref{sec:btree_construction} is much slower than our other methods.
The increase in memory consumption of our program in the multithreaded case is due to the use of thread-safe hash tables and multiple sequences and stacks residing in main memory.

For the Chromosome~19 collections and a single thread, \cmsbwt is the fastest program (at the cost of high memory usage), followed by \rpfbwt (for larger cases) and our algorithms.
For more threads, our program is always the fastest.
Regarding the memory consumption, \grlbwt uses the least amount of main memory for single-threaded processing, followed by our programs.
For the SARS-CoV-2 sequences, our program is the fastest, especially with multiple threads, and for more than one thread also the most memory efficient.

\section{Conclusion and Further Work}
\label{sec:conclusion}

We described an algorithm to compute the $\BBWT$---and by extension the common \$-$\BWT$ and various versions of the $\eBWT$---from the lexicographically sorted Lyndon grammar of a text or text collection.
Furthermore, we gave an algorithm that lexicographically sorts a Lyndon grammar and discussed approaches to efficiently compute the Lyndon grammar of a text or text collection.
We implemented the described algorithms and found that they outperform other current state-of-the art programs in terms of time or memory consumed (often both).


Currently, we use 32-bit integers for the grammar symbols.
For extremely large and diverse datasets, this may not suffice.

\bibliography{bib}



\end{document}